\documentclass[journal,twoside,web]{ieeecolor}

\usepackage{etoolbox}
\makeatletter
\@ifundefined{color@begingroup}%
{\let\color@begingroup\relax
\let\color@endgroup\relax}{}%
\def\fix@ieeecolor@hbox#1{%
\hbox{\color@begingroup#1\color@endgroup}}
\patchcmd\@makecaption{\hbox}{\fix@ieeecolor@hbox}{}{\FAILED}
\patchcmd\@makecaption{\hbox}{\fix@ieeecolor@hbox}{}{\FAILED}

\usepackage{generic}
\usepackage{cite}

\usepackage{hyperref}
\hypersetup{hidelinks=true}
\usepackage{textcomp}
\usepackage{enumerate}

\usepackage{amsmath,amssymb,amsfonts}
\usepackage{bm}
\usepackage{mathbbol,dsfont}
\usepackage{algorithmic,algorithm}
\usepackage{graphicx}
\usepackage{textcomp}
\usepackage{epstopdf}
\usepackage{subfigure}
\allowdisplaybreaks[4]
\makeatletter
\newenvironment{breakablealgorithm}
{
		\begin{center}
			\refstepcounter{algorithm}
			\hrule height.8pt depth0pt \kern2pt
			\renewcommand{\caption}[2][\relax]{
				{\raggedright\textbf{\ALG@name~\thealgorithm} ##2\par}%
				\ifx\relax##1\relax 
				\addcontentsline{loa}{algorithm}{\protect\numberline{\thealgorithm}##2}%
				\else 
				\addcontentsline{loa}{algorithm}{\protect\numberline{\thealgorithm}##1}%
				\fi
				\kern2pt\hrule\kern2pt
			}
		}{
		\kern2pt\hrule\relax
	\end{center}
}
\makeatother

\newtheorem{remark}{Remark}
\newtheorem{assumption}{Assumption}
\newtheorem{example}{Example}
\newtheorem{definition}{Definition}
\newtheorem{lemma}{Lemma}
\newtheorem{theorem}{Theorem}

\def\BibTeX{{\rm B\kern-.05em{\sc i\kern-.025em b}\kern-.08em
    T\kern-.1667em\lower.7ex\hbox{E}\kern-.125emX}}
\begin{document}
\title{One-bit consensus of controllable linear multi-agent systems with communication noises}
\author{Ru An, Ying Wang, \IEEEmembership{Member, IEEE}, Yanlong Zhao, \IEEEmembership{Senior Member, IEEE}, and Ji-Feng Zhang, \IEEEmembership{Fellow, IEEE}
\thanks{This work was supported by National Natural Science Foundation of China under Grants 62025306, 62303452, 62433020, and T2293770, CAS Project for Young Scientists in Basic Research under Grant YSBR-008, China Postdoctoral Science Foundation under Grant 2022M720159. \textit{(Corresponding author: Ying Wang.)}}
\thanks{Ru An and Yanlong Zhao are with the Key Laboratory of Systems and Control,  Academy of Mathematics and Systems Science, Chinese Academy of Sciences, Beijing 100190, P. R. China, and also with the School of Mathematics Sciences, University of Chinese Academy of Sciences, Beijing 100149, P. R. China.
(e-mail: anru@amss.ac.cn; ylzhao@amss.ac.cn).}
\thanks{Ying Wang is with the Key Laboratory of Systems and Control,  Academy of Mathematics and Systems Science, Chinese Academy of Sciences, Beijing 100190, P. R. China, and also with the Division of Decision and Control Systems, Royal Institute of Technology, Stockholm 11428, Sweden.
(e-mail: wangying96@amss.ac.cn).}
\thanks{Ji-Feng Zhang is with the School of Automation and Electrical Engineering, Zhongyuan University of Technology, Zhengzhou 450007, Henan Province, P. R. China, and also with the Key Laboratory of Systems and Control,  Academy of Mathematics and Systems Science, Chinese Academy of Sciences, Beijing 100190, P. R. China.
 (e-mail: jif@iss.ac.cn). 
}}

\maketitle

\begin{abstract}
This paper addresses the one-bit consensus of controllable linear multi-agent systems (MASs) with communication noises.
A consensus algorithm consisting of a communication protocol and a consensus controller is designed.
The communication protocol introduces a linear compression encoding function to achieve a one-bit data rate, thereby saving communication costs.
The consensus controller with a stabilization term and a consensus term is proposed to ensure the consensus of a potentially unstable but controllable MAS.
Specifically, in the consensus term, we adopt an estimation method to overcome the information loss caused by one-bit communications and a decay step to attenuate the effect of communication noise.
Two combined Lyapunov functions are constructed to overcome the difficulty arising from the coupling of the control and estimation.
By establishing similar iterative structures of these two functions, this paper shows that the MAS can achieve consensus in the mean square sense at the rate of the reciprocal of the iteration number under the case with a connected fixed topology.
Moreover, the theoretical results are generalized to the case with jointly connected Markovian switching topologies by establishing a certain equivalence relationship between the Markovian switching topologies and a fixed topology.
Two simulation examples are given to validate the algorithm.
\end{abstract}

\begin{IEEEkeywords}
consensus, one-bit data rate, communication noise,  controllable linear MASs, Markovian switching topologies
\end{IEEEkeywords}

\section{Introduction}
\label{sec:introduction}

\subsection{Background and motivation}
\IEEEPARstart{O}{ver} the past two decades, the consensus control of multi-agent systems (MASs) has been playing an increasingly important role in various fields, including engineering, communication, and biology.
For example, in the engineering field, consensus control plays a crucial role in applications such as attitude alignment of satellites, rendezvous in space, and cooperative control of unmanned aerial vehicles \cite{tanner-flocking-TAC2007,fang-distributedformation-TAC2023,tanner-stable-CDC2003}. 
In the communication field, it has been applied to problems like reputation consensus among mobile nodes \cite{liu-reputation-IEEEWCN2003} and load balancing in internet data centers \cite{loon-loadbalance-IEAC2016}. 
In the biology field, consensus mechanisms are essential in understanding phenomena such as the aggregation behavior of animals \cite{conradt-consensus-TEE2005} and the synchronous firing of biological oscillators \cite{mirollo-biologicaloscillators-SIAM-1990}.

With the increasing application of the consensus control across various fields, theoretical research on this topic has expanded significantly, such as \cite{ma-necessary-TAC2010,you-network-TAC2011,su-stability-TAC2011,gu-consensusability-TAC2011,wang-optimal-SCIS2024,ren-mas-ACC2005,huang-stochastic-ACC2008, li-mean-Auto2009, li-consensus-TAC2010, cheng-mean-TAC2014, wang-seeking-TAC2015, cheng-convergence-TAC2016, wang-consensus-IJSS2015}.
Consensus controllers are typically formulated as the sum of the state differences between an agent and its neighbors, with the addition of a step coefficient.
According to existing literature, the choice of the step coefficient usually depends on the MAS and significantly impacts the consensus rate. 
A constant step coefficient is often used to stabilize unstable systems \cite{ma-necessary-TAC2010,you-network-TAC2011,su-stability-TAC2011,gu-consensusability-TAC2011,wang-optimal-SCIS2024}.
Meanwhile, a decay step coefficient is commonly employed to attenuate the impact of noises, which is a widely adopted method in practice  \cite{ren-mas-ACC2005,huang-stochastic-ACC2008, li-mean-Auto2009, li-consensus-TAC2010, cheng-mean-TAC2014, wang-seeking-TAC2015, cheng-convergence-TAC2016, wang-consensus-IJSS2015}.

Due to the advantages of low communication costs and robustness, digital signals have become mainstream.
It is known that data communication generally consumes significantly more energy and incurs higher costs compared to data processing \cite{pottie-wireless-ACM2000}. 
These two factors make finite-bit data transmission between agents preferable and prevalent.
Motivated by the advantages and challenges associated with finite-bit data, consensus control with finite-bit communications has attracted increasing attention in various fields.
 
%

\subsection{Related literature}
In fact, significant research has been conducted on the consensus control of MASs with finite-bit communications. 
Quantizer plays a crucial role in converting accurate communications into finite-bit communications in practical applications. 
Consequently, numerous studies have investigated consensus control using different quantizers, such as integer quantizers, logarithmic quantizers, uniform quantizers, binary-valued quantizers, and others.
For example, Kashyap et al. in \cite{kashyap-quantized-Auto2007} and Chamie et al. in \cite{chamie-design-TAC2016} considered the consensus control with integer quantized communications.
Carli et al. in \cite{carli-communication-Auto2008} introduced a logarithmic quantizer in the consensus control to improve the control performance. 
Li et al. in \cite{li-distributed-TAC2012} proved that the MAS can achieve consensus with finite bits under a uniform quantizer in the noise-free case.
Meng et al. in \cite{meng-finite-SIAM2017} extended \cite{li-distributed-TAC2012} into a high-order system.
Moreover, due to the significant reduction in communication costs offered by the binary-valued quantizer, the consensus control based on binary-valued communications has gained considerable research attention \cite{zhao-consensus-TAC2019,wang-consensus-TAC2020,wang-consensus-IJRNC2020,an-consensus-IEEETCNS2024}.
To be specific, Zhao et al. in \cite{zhao-consensus-TAC2019} constructed a two-time-scale consensus algorithm and proved that the MAS can achieve mean square consensus with communication noises. 
Wang et al. in \cite{wang-consensus-TAC2020} proposed a consensus algorithm based on a recursive projection identification algorithm and obtained a mean square consensus rate, faster than that given by \cite{zhao-consensus-TAC2019}.
Wang et al. in \cite{wang-consensus-IJRNC2020} and An et al. in \cite{an-consensus-IEEETCNS2024} extended the system of \cite{wang-consensus-TAC2020} to the high-order MAS under fixed and switching topologies, respectively.
It is worth noting that the number of bits required for communication in the above consensus control depends not only on the choice of quantizer but also on the dimension of the agent's state. 
This implies that a one-bit data rate can be achieved only in first-order systems with binary-valued communications, 
as shown in \cite{zhao-consensus-TAC2019}--\!\!\cite{wang-consensus-TAC2020}, but the one-bit data rate cannot be achieved in \cite{kashyap-quantized-Auto2007, chamie-design-TAC2016, carli-communication-Auto2008, li-distributed-TAC2012, meng-finite-SIAM2017, wang-consensus-IJRNC2020, an-consensus-IEEETCNS2024}.

In a communication network, the connectivity between agents significantly impacts the system's cooperation effectiveness. Most existing consensus research focuses on the fixed topology (such as \cite{kashyap-quantized-Auto2007,chamie-design-TAC2016,carli-communication-Auto2008,li-distributed-TAC2012,zhao-consensus-TAC2019,wang-consensus-TAC2020}), while only a small portion addresses simpler cases with the switching topologies.
Even for the switching topologies,  certain restrictions remain.
For example, the communication network is assumed to be periodically connected in \cite{li-quantized-SCL2014} and \cite{meng-output-IJRNC2016}, and modeled as an i.i.d. process in \cite{an-consensus-IEEETCNS2024} and \cite{hu-consensus-SCIS2022}.
However, in practical applications, factors such as packet dropouts, environmental dynamics, link failures, and high-level scheduling commands lead to network topologies that switch with inherent correlations. 
Consequently, it becomes essential to model topology switching as a Markov process, which effectively captures the inherent correlations. 
Therefore, there is a need to study consensus control with finite-bit communications under both fixed topology and Markovian switching topologies.


\subsection{Main contribution}

In this paper, we consider the one-bit consensus of controllable linear MASs with communication noises for both fixed topology and Markovian switching topology cases.
The main contributions of this paper are as follows:

\begin{itemize}

\item The system model studied in this paper is the most general linear system model of consensus control with finite-bit communications, requiring only controllability.
Compared with previous studies of consensus with finite-bit communications \cite{wang-consensus-IJRNC2020}--\!\!\cite{an-consensus-IEEETCNS2024}, this paper removes the strict constraints of orthogonality and full row rank of the coefficient matrices, greatly expanding the applicability of the system model. 
Besides, this paper realizes the control of a high-order system with a first-order input, which simplifies the control process.
To the best of the author's knowledge, it is the first consensus study on high-order systems under one-bit communications.
\item 
A consensus algorithm consisting of a communication protocol and a consensus controller is proposed to achieve consensus with one-bit communications.
In the communication protocol, a linear compression encoding function is introduced to compress state vectors into scalars, achieving a one-bit data rate and reducing communication costs compared to \cite{meng-finite-SIAM2017, wang-consensus-IJRNC2020,an-consensus-IEEETCNS2024}.
The consensus controller includes a stabilization term to ensure the stability of MASs and a consensus term with a decay step to attenuate the effect of stochastic communication noises.
To overcome the information loss caused by the one-bit data rate, an estimation method is used in the consensus term to infer the neighbors' states from one-bit communications.
\item
The consensus properties of the algorithm are established under the connected fixed topology case.
Two combined Lyapunov functions are constructed to overcome the difficulty arising from the coupling of the control and estimation.
By establishing similar iterative structures of these two functions, this paper shows that the compressed states of MAS achieve consensus at a rate of $O(\frac{1}{t})$. 
Through establishing the consensus equivalence between the original and compressed states, it is shown that the MAS can achieve consensus in the mean square sense at a rate of $O(\frac{1}{t})$.

\item
The theoretical results are generalized to the case with jointly connected Markovian switching topologies by establishing a certain equivalence relationship between the Markovian switching topologies and a fixed topology.
To be specific, the MAS also can achieve consensus in the mean square sense at a rate of $O(\frac{1}{t})$ under jointly connected Markovian switching topologies with appropriate step coefficients of estimation and controller.
It is worth noting that the step coefficients depend on the switching probability of the Markovian switching topologies.


%

\end{itemize}

The remainder of this paper is organized as follows:
Section  \ref{sec2} gives the preliminaries of basic concepts and graph theory and describes the consensus problem. Section \ref{sec3}  introduces the consensus algorithm. The main results of this paper are presented in Section \ref{fixed}, which includes the main convergence and consensus results. 
Section \ref{switching} generalizes the theoretical results in Section \ref{fixed} to Markovian switching topologies.
Section \ref{sec5} gives two simulation examples for the fixed and switching topology case. Section \ref{sec-conclusion} is the summary and prospect of this paper.

\section{Preliminaries and problem formulation}\label{sec2}
In this section, we first give some basic concepts in matrix and graph theory, and subsequently formulate the system model and the consensus problems investigated in this paper.

\subsection{Basic concept}
Let $\mathds{R}$ denote the set of real numbers, and $\mathds{N}$ denote the set of natural numbers.
We use $x\in\mathds{R}^{n}$ and $A\in\mathds{R}^{n\times m}$ to denote $n$-dimensional column vector and $n\times m$-dimensional real matrix, respectively.
Denote $\vec{0}_{m}=[0,\ldots,0]^{T}\in \mathds{R}^{m}$ and $\vec{1}_{m}=[1,\ldots,1]^{T}\in \mathds{R}^{m}$, where the notation $T$ denotes the transpose operator.
Denote $|a|$ as the absolute value of a scalar.
Moreover, we denote $\|x\|=\|x\|_{2}$ and $\|A\|=\sqrt{\lambda_{\max}(A^{T}A)}$ as the Euclidean norm of vector and matrix, respectively,  where $\lambda_{\max}(\cdot) $ denotes the largest eigenvalue of the matrix.
Correspondingly, $\lambda_{\min}(\cdot)$ denotes the smallest eigenvalue of the matrix.
For symmetric matrices $A\in\mathds{R}^{m\times m}$ and $B\in \mathds{R}^{m\times m}$, $A\ge B$ represents that $A-B$ is a positive semi-definite matrix.
diag$\{\cdot\}$ denotes the block-diagonal matrix. 
And, for arbitrary matrices $A=[a_{ij}]\in \mathds{R}^{m\times n}$ and $B\in \mathds{R}^{p\times q}$, the Kronecker product of $A$ and $B$ is defined as
\begin{equation}\nonumber
A\otimes B \triangleq
\begin{bmatrix}
a_{11}B & a_{12}B & \cdots & a_{1n}B\\
a_{21}B & a_{22}B & \cdots & a_{2n}B\\
\vdots      & \vdots     &            &\vdots     \\
a_{m1}B & a_{m2}B & \cdots & a_{mn}B  
\end{bmatrix}
\in \mathds{R}^{mp\times nq}.
\end{equation}

In addition, the mathematical expectation is denoted as $E[\cdot]$. $\mathds{D}$ denotes the one-step forward operator, i.e., let $x(k)$ be a sequence of variables, then $\mathds{D}x(k)=x(k+1)$.

\subsection{Graph theory}
In order to describe the relation between agents, we introduce a topology $G=(N_{0}, E)$, where $N_{0}= \{1, \ldots, N\}$ is the set of agents, and  $E\subseteq N_{0}\times N_{0}$ is the ordered edges set of the topology $G$. 
Denote $N_{i}$  as the neighbor set of the agent $i$ in the topology $G$.  
Denote the adjacency matrix  of the $N$ agents as $A_{G}$, where each element of the matrix $A_{G}$ satisfies $a_{ij}=1$ if $(i,j)\in E$, else  $a_{ij}=0$. 
Denote the degree matrix of the $N$ agents as $D$, where $D=\mathrm{diag}\{d_{1}, d_{2}, \ldots, d_{N}\}$ and $d_{i}$ is the degree of agent $i$. Then, the Laplace matrix of $G$ is defined as $L = D - A_{G}$.
Denote $d_{\max}=\max_{1\le i\le N}{d_{i}}$ and $d=\sum_{i=1}^{N}d_{i}$.


\subsection{Problem formulation}

Consider the following MAS with $N$ agents at time $t$:
\begin{equation}\label{MAS}
x_{i}(t+1)=Ax_{i}(t)+Bu_{i}(t),\quad i=1, \ldots, N,
\end{equation}
where $A \in \mathds{R}^{n\times n}$ and $B\in \mathds{R}^{n} $ are constant matrices, $x_{i}(t)=[x_{i1}(t),\ldots,x_{in}(t)]^{T}\in \mathds{R}^{n}$ is the state of the agent $i$ at time $t$, and $u_{i}(t)\in \mathds{R}$ is the control input of the agent $i$ at time $t$.

Agent $i$ receives one-bit information affected by communication noise from its neighbor $j$:\begin{equation}\label{binary}
\begin{cases}
y_{ij}(t)=g(x_{j}(t))+d_{ij}(t),
\\
s_{ij}(t)=\mathbb{1}_{\{y_{ij}(t)\le c_{ij}\}},
\end{cases}
\end{equation}
where the agent $j$ is the neighbor of the agent $i$ at time $t$, 
$g(\cdot): \mathds{R}^{n}\to\mathds{R}$ is a compression encoding function to be designed,
$d_{ij}(t) \in \mathds{R}$ is the communicating noise, $y_{ij}(t) \in \mathds{R}$ is the virtual output, $c_{ij}\in \mathds{R}$ is the threshold value,  $s_{ij}(t)$ is the one-bit information that the agent $i$ collects from its neighbor $j$, $\mathds{1}_{\{ a \le c \}}$ is the indicator function  defined as:
\begin{equation}\nonumber
\mathbb{1}_{\{a\le c\}}=
\begin{cases}
1, \quad  a\le c,\\
0, \quad a>c.
\end{cases}
\end{equation}

\begin{remark}\label{compression-function}
The compression encoding function $g(\cdot)$ is a common tool to save communication costs in the communication field \cite{lotfi-compressed-IEEETSP2020, hayakawa-asymptotic-IEEETSP2022, kafle-noise-IEEETSP2024}, as it provides savings of scarce network resources such as communication bandwidth, transmit/processing power, and storage.
In contrast with \cite{meng-finite-SIAM2017, wang-consensus-IJRNC2020, an-consensus-IEEETCNS2024}, the compression encoding function $g(\cdot)$ maps $n$-dimensional vectors to scalars to realize one-bit data rate communication, but agents in \cite{meng-finite-SIAM2017, wang-consensus-IJRNC2020, an-consensus-IEEETCNS2024} need to transmit finite-bit data depending on the dimension of states.

\end{remark}

To proceed with our analysis, we introduce two assumptions about the system model and the communication noises.

\begin{assumption}\label{assm-AB}
The linear system $(A, B)$ is controllable.
\end{assumption}

\begin{assumption}\label{assm-d}
The noises $\{d_{ij}(t),i,j\in N_{0},t\in\mathds{N}\}$ are independent and identically distributed as $N(0,\delta^2)$ for indices $i,j$ and time $t$, with known distribution function $F(\cdot)$ and density function $f(\cdot)$.
\end{assumption}

\begin{remark}
Compared with the existing consensus works based on finite-bit communications, the condition of the system model in this paper is the weakest, requiring only controllability, as stated in Assumption \ref{assm-AB}.
To be specific, in \cite{meng-finite-SIAM2017}, besides the requirement of controllability for $(A, B)$, there were restrictions on the eigenvalues of the coefficient matrix $A$. 
 In \cite{wang-consensus-IJRNC2020}, orthogonality constraints on the coefficient matrices were required.
 In \cite{an-consensus-IEEETCNS2024}, the system model is assumed to be neutrally stable.
\end{remark}



\begin{definition}\label{def}
(\!\cite{cheng-convergence-TAC2016} \textbf{Consensus}).
Denote  $x_{i}(t)$ as the state of the agent $i$ at time $t$, where $i=1,\ldots,N$. 
For all agents, if $x_{i}(t)$, $i=1,\ldots,N,$ satisfy:

(1) $E[\|x_{i}(t)\|^{2}]<\infty$, $i=1,\ldots,N$;

(2) $\lim_{t\to \infty}E[\|x_{i}(t)-\bar{x}(t)\|^{2}]=0$,  $i\in\{1,\ldots, N\}$, where $\bar{x}(t)=\frac{1}{N}\sum_{i=1}^{N}x_{i}(t)$.

Then, the agents are said to achieve consensus. 
\end{definition}

\textit{Problem:} The goal of this paper is to design a controller $u_{i}(t)$ and the communication mechanism $g(\cdot)$ based on one-bit communications $s_{ij}(t)$ to achieve consensus of the controllable MAS \eqref{MAS}-\eqref{binary}.

\section{Algorithm design}\label{sec3}

This section focuses on designing a consensus algorithm that enables linear systems to reach a consensus with communication noises and a one-bit data rate constraint. 

In order to simplify the design of the algorithm, this paper considers the MAS \eqref{MAS}-\eqref{binary} in  the Brunovsky canonical form, where 
\begin{equation}\label{canonical}
A=\tilde{A}\triangleq
\begin{bmatrix}
0 &1 &\cdots &0\\
\vdots & \vdots &\ddots &\vdots\\
0 & 0 & \cdots & 1\\
a_{1} & a_{2}&\cdots &a_{n}
\end{bmatrix}\in \mathds{R}^{n\times n},
B=\tilde{B}\triangleq
\begin{bmatrix}
0\\
\vdots\\
0\\
1
\end{bmatrix}\in \mathds{R}^{n}.
\end{equation}

\begin{remark}\label{rem-can}
It can be seen that for any controllable system \eqref{MAS}, there exists a nonsingular matrix $P$ that can transform \eqref{MAS} into this Brunovsky canonical form \cite{brunovsky-classification-K1970}, i.e.,
 $PAP^{-1}=\tilde{A}$ and $PB=\tilde{B}$. Let $\tilde{x}_{i}(t)=Px_{i}(t)$. Then, \eqref{MAS} is transformed into $\tilde{x}_{i}(t+1)=\tilde{A}\tilde{x}_{i}(t)+\tilde{B}u_{i}(t)$, which is equivalent to \eqref{MAS}-\eqref{canonical}.
\end{remark}

Next, we introduce the design idea of the consensus algorithm.
The algorithm consists of a communication protocol and a consensus controller.
First, a communication protocol including a simple and effective linear compression encoding function is developed, greatly saving communication costs. 
Then,  inspired by \cite{cheng-mean-TAC2014, wang-seeking-TAC2015, cheng-convergence-TAC2016, wang-consensus-IJSS2015}, a controller is designed with a stabilization term ensuring the stability of MASs and a consensus term reducing state differences between agents.
Additionally, in the consensus term, an estimation method is used to infer the neighbors' states from one-bit communications,  and a decay step is adopted to attenuate the effect of noise. 

Based on the above idea, we propose a consensus algorithm involving both communication protocol and consensus controller in Algorithm 1.
%

\begin{breakablealgorithm}
\begin{algorithmic}
\caption{}
\label{algorithm1}

\textbf{Initiation:} Denote the integer $t_{0}(>0)$ as the initial time. $x_{i}(t_{0}+1)=x_{i}^{0}$ is the initial state of the agent $i$, $\hat{z}_{ij}(t_{0})=\hat{z}_{ij}^{0}$ is the initial estimate of the agent $j$  estimated by the agent $i$. 
 Then, for $t\ge t_{0}+1,$ the algorithm is as follows.

\textbf{Step 1: Communication protocol:}

Denote the compression encoding function $g(x_{j}(t))=K_{2}x_{j}(t)$, where $K_{2}=[b_{1}, b_{2}, \ldots,b_{n-1},1]\in \mathds{R}^{1\times n}$, where $b_{1}$, $b_{2}$, $\ldots$, $b_{n-1}$ are the compression coefficients to be designed.

\textbf{Step 2: Consensus controller:}

 \textbf{Step 2.1 Estimation:} each agent $i$ estimates the compressed state $K_{2}x_{j}(t)$ of its neighbor agent $j$ at time $t$ by
\begin{equation}\label{estimate}
\hat{z}_{ij}(t)=\Pi_{M}\Big\{ \hat{z}_{ij}(t-1)+\frac{\beta}{t}\Big(F\big(c_{ij}-\hat{z}_{ij}(t-1)\big)-s_{ij}(t)\Big) \Big\},
\end{equation}
where $j\in N_{i}$, $\beta$ is the step coefficient for estimation updating, $\Pi_{M}(\cdot)$ is a projection mapping defined as
\begin{equation}\label{piM}
\Pi_{M}(\zeta)=\mathop{\arg\min}\limits_{|\xi|\le M}| \zeta - \xi |,\forall \zeta \in \mathds{R},
\end{equation}
where $M$ is the upper bound of $|K_{2}x_{i}^{0}|$ and $|\hat{z}_{ij}^{0}|$.

 \textbf{Step 2.2 Controller:} based on these estimates, each agent $i$ designs its control by
\begin{equation}\label{control}
u_i(t)=K_{1}x_{i}(t)+\frac{\gamma}{t+1}\sum_{j \in N_{i}}\big(\hat{z}_{ij}(t)-K_{2}x_{i}(t)\big),
\end{equation}
where $K_{1}=[-a_{1}+b_{1}, -a_{2}+b_{2}-b_{1}, \ldots, -a_{n-1}+b_{n-1}-b_{n-2}, -a_{n}-b_{n-1}+1]\in \mathds{R}^{1\times n}$ and
$\gamma$ is the step coefficient of the controller that needs to be designed.

\end{algorithmic}
\end{breakablealgorithm}

\begin{remark}\label{g-difficulty}
It is worth noting that although the compression encoding function saves communication costs, it complicates the recovery of the original states $x_{j}(t)$ from the compressed states $K_{2}x_{j}(t)$.
Since the compression encoding function $g(x_{j}(t))=K_{2}x_{j}(t)$ is linear, the compressed states $K_{2}x_{j}(t)$ are irreversible. 
The key issue to be explored is how to design the linear compression coefficient $K_{2}$ in the communication protocol to ensure that the consensus of original and compressed states is equivalent, thereby achieving consensus of the original states through the consensus of the compressed states.
\end{remark}

\begin{remark}
In the first step of the consensus controller in Algorithm \ref{algorithm1}, the upper bound is a kind of common global information determined by the compressed states' initial values and estimates.
For example, denote $M_{1}$ as the upper bound for the norm of these initial values, i.e., $M_{1}\ge\|x_{i}^{0}\| $, $M_{1}\ge |z_{ij}^{0}|$. Then, denote $M=(1+\sum_{i=1}^{N}|b_{i}|)M_{1}$, it is obviously that
$|K_{2}x_{i}^{0}|\le M$ and $ |z_{ij}^{0}|\le M$.
\end{remark}

\begin{remark}\label{projection}
The estimation procedure adopted in Algorithm \ref{algorithm1} is widely used in the identification problems with one-bit data, which is a recursive projection identification algorithm, such as \cite{wang-consensus-TAC2020, wang-consensus-IJRNC2020, an-consensus-IEEETCNS2024}, and \cite{guo-recursive-Auto2013}. 
The projection operator $\Pi_{M}$ is used to ensure the boundedness of the estimates and compressed states to obtain the scaling factor for convergence analysis.
Besides, as \cite[Proposition 6]{guo-recursive-Auto2013} points out, the projection mapping given by \eqref{piM} has the following property:
$$| \Pi_{M}(x_{1})-\Pi_{M}(x_{2})| \le | x_{1}-x_{2}|,\forall x_{1},x_{2}\in \mathds{R}.$$
\end{remark}

\begin{remark}\label{K1K2}
Algorithm  \ref{algorithm1} is applicable not only to Brunovsky canonical form but also to general controllable systems, which requires only a transformation in the control gain $K_{1}$ and compression coefficient $K_{2}$.
Specifically, $K_{1}P$ and $K_{2}P$ are used as the control gain and the compression coefficient for general controllable systems, where $P$ is the transformation matrix in Remark \ref{rem-can}.
%

Besides, control gains $K_{1}$ and $K_{2}$ designed in Algorithm \ref{algorithm1} has the property that $K_{2}(A+BK_{1})=K_{2}$ and $K_{2}B=1.$
\end{remark}

By  \eqref{MAS} and \eqref{control}, the state $x_{i}(t)$ is updated as 
\begin{equation}\label{update}
x_{i}(t+1)=(A+BK_{1})x_{i}(t)+\frac{\gamma B}{t+1}\sum_{j \in N_{i}}(\hat{z}_{ij}(t)-K_{2}x_{i}(t)).
\end{equation}

To better illustrate the convergence analysis of Algorithm \ref{algorithm1}, we first consider the case with fixed topology.
\begin{assumption}\label{assm-G}
The topology graph $G$  is connected. 
\end{assumption}

For the convenience of the subsequent analysis, we rewrite the estimation and update formulas in vector form based on the fixed topology $G$. 

Firstly, define $x(t)=[x_{1}^{T}(t),x_{2}^{T}(t),\ldots, x_{N}^{T}(t)]^{T}\in \mathds{R}^{nN}$.
Then, denote 
\begin{align*}
\hat{z}(t)=&[\hat{z}_{1r_{1}}(t),\hat{z}_{1r_{2}}(t),\ldots,\hat{z}_{1r_{d_{1}}}(t),\ldots,\hat{z}_{ir_{d_{1}+\cdots+d_{i-1}+1}}(t),
\\
&\ldots,\hat{z}_{ir_{d_{1}+\cdots+d_{i}}}(t),\ldots,\hat{z}_{Nr_{d_{1}+\cdots+d_{N}}}(t)]^{T}\in \mathds{R}^{d},
\end{align*}
where $r_{d_{1}+d_{2}+\cdots+d_{i-1}+1}, \ldots, r_{d_{1}+\cdots+d_{i}}\in N_{i}$ for $i=1,2,$
$\ldots,N$.
Similarly, denote 
\begin{align*}
s(t)=&[s_{1r_{1}}(t),s_{1r_{2}}(t),\ldots,s_{1r_{d_{1}}}(t),\ldots,s_{ir_{d_{1}+\cdots+d_{i-1}+1}}(t),
\\
&\ldots,s_{ir_{d_{1}+\cdots+d_{i}}}(t),\ldots, s_{Nr_{d_{1}+\cdots+d_{N}}}(t)]^{T}\in \mathds{R}^{d},
\end{align*}
and
\begin{align*}
C=&[c_{1r_{1}}, c_{1r_{2}},\ldots, c_{1r_{d_{1}}}, \ldots, c_{ir_{d_{1}+\cdots+d_{i-1}+1}},\ldots, 
\\
&c_{ir_{d_{1}+\cdots+d_{i}}},\ldots,c_{Nr_{d_{1}+\cdots+d_{N}}}]^{T}\in \mathds{R}^{d}.
\end{align*}

Without loss of generality, 
assume  the subscript $r_{s}$ in vector $\hat{z}(t)$ represents the neighbor $j$ of the agent $i$, i.e., $\hat{z}_{ir_{s}}(t)=\hat{z}_{ij}(t)$,
where $r_{s}\in N_{i}$, $s\in \{d_{1}+d_{2}+\cdots+d_{i-1}+1,\ldots, d_{1}+\cdots+d_{i}\}.$
Based on the above notation, we construct two matrices to establish the relation between agents' compressed states and their estimates.


$Q$ is designed to select the true compressed state of the agent that correlates with its estimate. 
Define $Q=[Q_{1r_{1}},\ldots,Q_{1r_{d_{1}}},$
$\ldots,Q_{Nr_{d_{1}+\cdots+d_{N-1}+1}},\ldots,Q_{Nr_{d_{1}+\cdots+d_{N}}}]^{T}\in \mathds{R}^{d \times N}$,
where
$Q_{ir_{s}}=Q_{ij}=[\vec{0}_{j-1}^{T},1,\vec{0}_{N-j}^{T}]^{T}\in \mathds{R}^{ N}$ for $(i,r_{s})\in E$, else $Q_{ir_{s}}=\vec{0}_{N}$.

$W$ is designed to select the neighbor set of each agent. 
Define $W=[W_{1},\ldots, W_{N}]^{T}\in \mathds{R}^{N \times d},$ where
$W_{i}=[\vec{0}^{T}_{d_{1}+\cdots+d_{i-1}}, \vec{1}^{T}_{d_{i}},\vec{0}^{T}_{d_{i+1}+\cdots+d_{N}}]^{T}\in \mathds{R}^{d}$ for  $i\in\{1,\ldots,N \}$,

Based on the above matrices, the vector forms of estimation and update are given as follows:

$1.$Estimation:
\begin{align}
\hat{z}(t)=&\bm{\Pi_{M}}\Big\{\hat{z}(t-1)+\frac{\beta}{t}\big(\mathcal{F}(C-\hat{z}(t-1))-s(t)\big)\Big\},
\label{Estimate}
\end{align}
where $\bm{\Pi_{M}}(z)=[\Pi_{M}(z_{1}),\ldots, \Pi_{M}(z_{d})]^{T}$, $\mathcal{F}(z)=[F(z_{1}),$
$\ldots,F(z_{d})]^{T}$, for any $z=[z_{1},z_{2},\ldots, z_{d}]^{T}\in \mathds{R}^{d}$.

$2.$Update:
\begin{align}\label{Update}
x(t)={}&\Big(I_{N}\otimes (A+BK_{1})-\frac{\gamma}{t}L\otimes BK_{2}\Big)x(t-1)
\notag
\\
&+\frac{\gamma}{t}(W\otimes B)\hat{\varepsilon}(t),
\end{align}
where  $\hat{\varepsilon}(t)=\hat{z}(t)-(Q\otimes K_{2})x(t)$ is the estimation error of the compressed state $(I_{N}\otimes K_{2})x(t)$.

\section{Main resluts}\label{fixed}
In this section, we demonstrate that all agents can reach consensus and provide the corresponding consensus rate. 

At first, we analyze the consensus property of the compressed states and the convergence property of their estimates. 
Then, the equivalence between the consensus of the compressed and original states is established. 
Finally, the consensus of original states and the corresponding consensus rate are obtained.


\subsection{Properties of compressed states and their estimates}\label{fix-b}

In order to analyze the consensus of the compressed states and the convergence of their estimates,
we give the following lemmas first.

\begin{lemma}\label{L}
(\!\!\cite{hu-consensus-SCIS2022}).  If Assumption \ref{assm-G} holds, then matrix $L$ has these properties:
\begin{enumerate}[i)]
\item $L$  is a nonnegative definite matrix with rank $n-1$ and eigenvalues $0=\lambda_{1}<\lambda_{2}\le \lambda_{3}\le \ldots \le \lambda_{N}$;
\item There exists an orthogonal matrix $T_{G}$ such that $T_{G}^{-1}LT_{G}=\text{diag}\{\lambda_{1},\ldots,\lambda_{N}\}$.
\end{enumerate}
\end{lemma}

\begin{lemma}\label{bounded}
Under Assumptions \ref{assm-AB}-\ref{assm-G} and Algorithm \ref{algorithm1}, the compressed states $K_{2}x_{i}(t)$ and their estimates $\hat{z}_{ij}(t)$ are all bounded, i.e., $|K_{2}x_{i}(t)|\le M$ and $|\hat{z}_{ij}(t)|\le M$, where $M$ is defined as \eqref{estimate},   $i=1,2,\ldots, N$, $j\in N_{i}$, $t\ge t_{0}+1$.
\end{lemma}
\begin{proof}
First, due to the definition of $M$, we can get $|K_{2}x_{i}^{0}|\le M, |z_{ij}^{0}|\le M$.  By \eqref{estimate} and \eqref{piM}, we have $|\hat{z}_{ij}(t)|\le M$ for $t\ge t_{0}+1$.

Assume that $|K_{2}x_{i}(k)|\le M$ for $k=t_{0}+1,t_{0}+2,\ldots, t$.
By \eqref{MAS}, \eqref{update}, and Remark \ref{K1K2}, we have
%
\begin{align*}
&| K_{2}x_{i}(t+1)|
\\
=&\Big|K_{2}(A+BK_{1})x_{i}(t)+\frac{\gamma K_{2}B}{t+1}\sum_{j\in N_{i}}\big(\hat{z}_{ij}(t)-K_{2}x_{i}(t)\big)\Big|
\\
=&\Big|\left(1-\frac{d_{i}\gamma}{t+1}\right)K_{2}x_{i}(t)+\frac{\gamma }{t+1}\sum_{j\in N_{i}}\hat{z}_{ij}(t)\Big|.
\\
\le & \Big|1-\frac{\gamma d_{i}}{t+1}\Big|M+\frac{\gamma d_{i}}{t+1}M
\end{align*}

 For any given $\gamma$, we can choose an initial time $t_{0}$ that satisfies $t_{0}>\gamma d_{i}-1$. Since $1-\frac{\gamma d_{i}}{t+1}+\frac{\gamma d_{i}}{t+1}=1$,
we can get
$$
| K_{2}x_{i}(t+1)|\le\Big|1-\frac{\gamma d_{i}}{t+1}\Big|M+\frac{\gamma d_{i}}{t+1}M=M.
$$

Thus, by mathematical induction, we have $|K_{2}x_{i}(t)|\le M$ for all $t\ge t_{0}+1$. The lemma is proved.
\end{proof}

Next, due to the coupled relation between the control and estimation process, we introduce two Lyapunov functions, $V(t)$ and $R(t)$, to jointly analyze the consensus of the compressed states and the convergence of their estimates. 
These functions are defined as follows:
\begin{align}
\label{defineV}
V(t)&=E[\|(T_{G}^{-1}\otimes K_{2})\delta(t)\|^{2}],\\
\label{defineR}
R(t)&=E[\|\hat{\varepsilon}(t)\|^{2}],
\end{align}
where $\delta(t)=(J_{N}\otimes I_{n})x(t)$ and $J_{N}=I_{N}-\frac{1}{N}\vec{1}_{N}\vec{1}_{N}^{T}$.
Then, the following two lemmas show the coupled expressions of the two Lyapunov functions.

\begin{lemma}\label{V}
Under Assumptions \ref{assm-AB}-\ref{assm-G}, $V(t)$ satisfies
\begin{equation}
V(t)\le \Big(1-\frac{\gamma\lambda_{2}}{t}\Big)V(t-1)+\frac{\gamma\lambda_{G}}{t\lambda_{2}}R(t-1)+O\Big(\frac{1}{t^{2}}\Big),
\end{equation}
where $\lambda_{G}=\|J_{N}W\|^{2}$. 
\end{lemma}
\begin{proof}
See Appendix \ref{app-V}.
\end{proof}

\begin{lemma}\label{R}
Under Assumptions \ref{assm-AB}-\ref{assm-G}, $R(t)$ satisfies
\begin{align}
R(t)\le  \Big(1-\frac{2\beta f_{M}-\gamma\alpha}{t}\Big)R(t-1)+\frac{\gamma\lambda_{G}}{t\lambda_{2}}V(t-1)+O\Big(\frac{1}{t^{2}}\Big),
\end{align}
where $\alpha=2\sqrt{\lambda_{QW}}+\lambda_{QL}\lambda_{2}/\lambda_{G}$, $f_{M}=\min_{i,j\in N_{0}}\{f(|c_{ij}|+M)\}$, $\lambda_{QW}=\|QW\|^{2}$, and $\lambda_{QL}=\|QLT_{G}\|^{2}$. 
\end{lemma}
\begin{proof}
See Appendix \ref{app-R}.
\end{proof}

To establish the convergence properties of these two coupled functions $V(t)$ and $R(t)$, denote a new function $Z(t)=(V(t), R(t))^{T}$. Then, under Assumptions \ref{assm-AB}--\ref{assm-G}, we have
$$\|Z(t)\|\le\|(I-\frac{1}{t}U)Z(t-1)\|+O\Big(\frac{1}{t^2}\Big),$$
where 
$U=\begin{bmatrix}
u_{1}&u_{2}\\
u_{2}&u_{4}
\end{bmatrix}$,
$u_{1}=\gamma\lambda_{2}$, $u_{2}=\gamma\lambda_{G}/\lambda_{2}$, 
$u_{4}=2\beta f_{M}-\gamma\alpha$, $\alpha$ is the same as in Lemma \ref{R}.

\begin{lemma}\label{Z} (\!\!\cite{hu-consensus-SCIS2022}).  If Assumptions \ref{assm-AB}-\ref{assm-G} hold, 
then
\begin{equation}\nonumber
\|Z(t)\|=
\begin{cases}
O\left(\frac{1}{t^{\lambda_{\min}(U)}}\right),\quad &\lambda_{\min}(U)<1;
\\
O\left(\frac{\ln t}{t}\right),&\lambda_{\min}(U)=1;
\\
O\left(\frac{1}{t}\right),&\lambda_{\min}(U)>1.
\end{cases}
\end{equation}
\end{lemma}

\begin{remark}\label{vrz}
The analysis of $V(t)$ and $R(t)$ can be transformed into analyzing the convergence of $Z(t)$, it is because $0\le V(t)\le\|Z(t)\|$ and $0\le R(t)\le\|Z(t)\|$.
\end{remark}

\begin{theorem}\label{compressed}
Under Assumptions \ref{assm-AB}-\ref{assm-G}, the compressed states $K_{2}x_{j}(t)$ and their estimates $\hat{z}_{ij}(t)$ satisfy:
\begin{enumerate}[i)]
\item If $\beta>\frac{1}{2f_{M}}(\frac{\gamma \lambda_{G}^{2}}{\lambda_{2}^{3}}+\gamma \alpha)$, the compressed states reach consensus and their estimates converge to the real compressed states, i.e., for $i=1,\ldots,N$, $j\in N_{i}$,
$$
\lim_{t\to\infty}E[\|K_{2}x_{i}(t)-K_{2}\bar{x}(t)\|^{2}]=0,
$$
$$
\lim_{t\to\infty}E[\|\hat{z}_{ij}(t)-K_{2}x_{j}(t)\|^{2}]=0;
$$
\item If $\beta>\frac{1}{2f_{M}}(\frac{\gamma^{2}\lambda_{G}^{2}}{\lambda_{2}^{2}(\gamma \lambda_{2}-1)}+\gamma\alpha+1)$ and $\gamma>\frac{1}{\lambda_{2}}$, the compressed states reach consensus at the rate of $O(\frac{1}{t})$, and their estimates converge to the real compressed states at the rate of $O(\frac{1}{t})$, i.e., for $i=1,\ldots,N$, $j\in N_{i}$,
$$
E[\|K_{2}x_{i}(t)-K_{2}\bar{x}(t)\|^{2}]=O\Big(\frac{1}{t}\Big),
$$
$$
E[\|\hat{z}_{ij}(t)-K_{2}x_{j}(t)\|^{2}]=O\Big(\frac{1}{t}\Big),
$$
\end{enumerate}
where $\bar{x}(t)=\frac{1}{N}\sum_{i=1}^{N}x_{i}(t)$.
\end{theorem}
\begin{proof}
i) Let $|\lambda I_{2}-U|=\lambda^{2}-(u_{1}-u_{4})^{2}+4u_{2}^{2}=0$. Then, we have 
$$\lambda_{\min}(U)=\frac{1}{2}\Big(u_{1}+u_{4}-\sqrt{(u_{1}+u_{4})^{2}-4(u_{1}u_{4}-u_{2}^{2})} \Big).$$

If $\beta>\frac{1}{2f_{M}}(\frac{\gamma \lambda_{G}^{2}}{\lambda_{2}^{3}}+\gamma \alpha)$, then
$u_{1}u_{4}>u_{2}^{2}$. Since $u_{1}>0$ and $u_{1}u_{4}>u_{2}^{2}$, we have $\lambda_{\min}(U)>0$.

By Lemma \ref{Z}, if $\lambda_{\min}(U)>0$, then $\lim_{t\to\infty}\|Z(t)\|=0$. Hence,  due to the relation among $V(t)$, $R(t)$ and $Z(t)$ in Remark \ref{vrz}, it is clearly that
\begin{equation}\label{limVR}
\lim_{t\to \infty}V(t)=0,\quad \lim_{t\to \infty}R(t)=0,
\notag
\end{equation}
which is equivalent to $\lim_{t\to\infty}E[\|(T_{G}^{-1}\otimes K_{2})\delta(t)\|^{2}]=0$ and $\lim_{t\to\infty}E[\|\hat{\varepsilon}(t)\|^{2}]=0$.

Subsequently, there is $\lim_{t\to\infty}E[\|I_{N}\otimes K_{2})\delta(t)\|^{2}]=0$, which implies that $\lim_{t\to\infty}E[\|K_{2}\delta_{i}(t)\|^{2}]=0$ for $i=1,\ldots,N$.
Therefore, we have that for $i=1,\ldots,N$, $j\in  N_{i}$,
$$
\lim_{t\to\infty}E[\|K_{2}x_{i}(t)-K_{2}\bar{x}(t)\|^{2}]=0,
$$ $$
\lim_{t\to\infty}E[\|\hat{z}_{ij}(t)-K_{2}x_{j}(t)\|^{2}]=0.
$$

ii) Similarly, if $\beta>\frac{1}{2f_{M}}(\frac{\gamma^{2}\lambda_{G}^{2}}{\lambda_{2}^{2}(\gamma \lambda_{2}-1)}+\gamma\alpha+1)$, then 
$$
\frac{u_{2}^{2}}{u_{1}-1}+1<u_{4}.
$$ 
If $\gamma>\frac{1}{\lambda_{2}}$, we have $u_{1}=\gamma\lambda_{2}>1$, then
$$
u_{2}^{2}+u_{1}-1<u_{4}(u_{1}-1).
$$ 

Subsequently, since $(u_{1}+u_{4})^{2}-4(u_{1}u_{4}-u_{2}^{2})-(u_{1}+u_{4}-2)^{2}=4(u_{2}^{2}+u_{1}-1-u_{4}(u_{1}-1))<0$, we have 
$$
u_{1}+u_{4}-\sqrt{(u_{1}+u_{4})^{2}-4(u_{1}u_{4}-u_{2}^{2})}>2,
$$
then $\lambda_{\min}(U)>1$.

By Lemma \ref{Z}, we have $\|Z(t)\|=O(\frac{1}{t})$, i.e., 
$$
V(t)=O\left(\frac{1}{t}\right),\quad R(t)=O\left(\frac{1}{t}\right).
$$
Similarly as the proof of Part i), we have 
$$
E[\|K_{2}x_{i}(t)-K_{2}\bar{x}(t)\|^{2}]=O\Big(\frac{1}{t}\Big),
$$
$$
E[\|\hat{z}_{ij}(t)-K_{2}x_{j}(t)\|^{2}]=O\Big(\frac{1}{t}\Big).
$$
where $i=1,\ldots,N$, $j\in  N_{i}$.
\end{proof}

\subsection{Consensus of the original states}
In order to establish the equivalence of the consensus of the agent states before and after compression, we need the compression coefficients to satisfy the following condition.
\begin{assumption}\label{assm-K}
The compression coefficients $b_{1}, b_{2},\ldots, $
$b_{n-1}$ satisfy that: All the roots $r_{1}, r_{2}, \ldots, r_{n-1}$ of $s^{n-1}+b_{n-1}s^{n-2}+\cdots+b_{2}s+b_{1}=0$ are inside the unit circle.
\end{assumption}

\begin{remark}\label{assm4}
The controllable system $(A, B)$ can be transformed into a neutrally stable system $(A+BK_{1}, B)$ if the stabilization term $K_{1}x_{i}(t)$ of the consensus controller satisfies Assumption \ref{assm-K}.
Subsequently, the design of the consensus term follows a similar approach to the controller design in a neutrally stable system \cite{an-consensus-IEEETCNS2024}.
Specifically, based on $\text{det}(sI_{n}-(A+BK_{1}))=(s-1)(s^{n-1}+b_{n-1}s^{n-2}+\cdots+b_{2}s+b_{1})$, the eigenvalues of  $A+BK_{1}$ are $1, r_{1},\ldots, r_{n-1}$.
Under Assumption \ref{assm-K},  $(A+BK_{1}, B)$ is neutrally stable, i.e.,  all the eigenvalues of $A+BK_{1}$ are not outside the unit circle.
%
Thus, the system with coefficient $(A+BK_{1}, B)$ can be controlled by a consensus term that has a decay step, which is a common tool to attenuate the effect of the noises in first-order systems \cite{ren-mas-ACC2005, huang-stochastic-ACC2008, li-mean-Auto2009, li-consensus-TAC2010,zhao-consensus-TAC2019, wang-consensus-TAC2020}. This approach also applies to orthogonal systems \cite{wang-consensus-IJRNC2020} and neutrally stable systems but with the limitation that $B$ is of full row rank \cite{an-consensus-IEEETCNS2024}.
\end{remark}

Then, two lemmas are given to deal with the relation between the compressed states $K_{2}x_{i}(t)$ and the original states $x_{i}(t)$. 
\begin{lemma}\label{kx-x}
The linear systems \eqref{MAS} with Brunovsky canonical form \eqref{canonical} satisfies
$$
\mathds{D}^{n-1}x_{in}(t)+b_{n-1}\mathds{D}^{n-2}x_{in}(t)+\cdots+b_{1}x_{in}(t)=\mathds{D}^{n-1}K_{2}x_{i}(t),
$$
where $x_{i}(t)=[x_{i1}(t),\ldots,x_{in}(t)]^{T}\in \mathds{R}^{n}$ and $i$$=$$1$, $2$, $\ldots$, $N$.
\end{lemma}
\begin{proof}
See Appendix \ref{app-kx}.
\end{proof}

\begin{lemma}\label{difference-equation} 
Consider the stochastic process $\xi(t)\in \mathds{R}$ that satisfies the following stochastic difference equation:
$$
\mathds{D}^{n-1}\xi(t)+b_{n-1}\mathds{D}^{n-2}\xi(t)+\cdots+b_{2}\mathds{D}\xi(t)+b_{1}\xi(t)=\eta(t),
$$
where $\eta(t)\in\mathds{R}$ is a stochastic process which converges to a finite-second-moment random variable $\eta^{*}$ in the mean square.
Under Assumption \ref{assm-K},
\begin{enumerate}[i)]
\item $\lim_{k\to\infty}E[|\xi(t)-\xi^{*}|^{2}]=0$;
\item If $E[|\eta(t)-\eta^{*}|^{2}]=O(\frac{1}{t})$, then  $E[|\xi(t)-\xi^{*}|^{2}]=O(\frac{1}{t})$, 
\end{enumerate}
where $\xi^{*}=\frac{1}{\prod_{j=1}^{n-1}(1-r_{j})}\eta^{*}$ and $r_{1}, \ldots, r_{n-1}$ are defined as Assumption \ref{assm-K}.
\end{lemma}
\begin{proof}
See Appendix \ref{app-de}.
\end{proof}

\begin{remark}\label{equivalence}
Lemmas \ref{kx-x}-\ref{difference-equation} demonstrate that the consensus of original states is equivalent to that of the compressed states.
Specifically, Lemma \ref{kx-x} shows that $x_{in}(t)$ and $K_{2}x_{i}(t)$ satisfy the difference equation in Lemma \ref{difference-equation} for each agent $i$. 
Since $x_{i}(t)-\bar{x}(t)$ is a linear combination of $x_{1}(t),\ldots,x_{N}(t)$, it can be seen that $x_{in}(t)-\bar{x}_{n}(t)$ and $K_{2}(x_{i}(t)-\bar{x}(t))$ still satify the difference equation in Lemma \ref{difference-equation}, where $\bar{x}(t)=[\bar{x}_{1}(t),\ldots,\bar{x}_{n}(t)]^{T}\in\mathds{R}^{n}$.
Then, the equivalence between the convergence of the original states $x_{in}(t)-\bar{x}_{n}(t)$ and compressed states $K_{2}(x_{i}(t)-\bar{x}(t))$ is given in Lemma \ref{difference-equation}.
\end{remark}

Using the convergence equivalence in Remark \ref{equivalence} and the results of the compressed states in Theorem \ref{compressed}, the consensus of MAS is obtained as follows.

\begin{theorem}\label{consensus}
Under Assumptions \ref{assm-AB}-\ref{assm-K}, the following results are obtained:
\begin{enumerate}[i)]
\item The MAS \eqref{MAS}-\eqref{canonical} reaches consensus, i.e.,
$$
\lim_{t\to\infty}E[\|x_{i}(t)-\bar{x}(t)\|^{2}]=0,
$$
providing with $\beta>\frac{1}{2f_{M}}(\frac{\gamma \lambda_{G}^{2}}{\lambda_{2}^{3}}+\gamma \alpha)$;

\item The MAS \eqref{MAS}-\eqref{canonical} reaches consensus at the rate of $O(\frac{1}{t})$, i.e.,
$$
E[\|x_{i}(t)-\bar{x}(t)\|^{2}]=O\left(\frac{1}{t}\right),
$$
if $\beta>\frac{1}{2f_{M}}(\frac{\gamma^{2}\lambda_{G}^{2}}{\lambda_{2}^{2}(\gamma \lambda_{2}-1)}+\gamma\alpha+1)$ and $\gamma>\frac{1}{\lambda_{2}}$.
\end{enumerate}
\end{theorem}
\begin{proof}
i) 
By Lemmas \ref{kx-x}--\ref{difference-equation} and Remark \ref{equivalence}, 
we have 
\begin{align}\label{x_in}
&\mathds{D}^{n-1}\big(x_{in}(t)-\bar{x}_{n}(t)\big)+b_{n-1}\mathds{D}^{n-2}\big(x_{in}(t)-\bar{x}_{n}(t)\big)\notag
\\
&+\cdots+b_{1}\big(x_{in}(t)-\bar{x}_{n}(t)\big)\notag
\\
=&\mathds{D}^{n-1}K_{2}\big(x_{i}(t)-\bar{x}(t)\big).
\end{align}

By Theorem \ref{compressed}, when $\beta>\frac{1}{2f_{M}}(\frac{\gamma \lambda_{G}^{2}}{\lambda_{2}^{3}}+\gamma \alpha)$, there is $\lim_{t\to\infty}E[\|K_{2}x_{i}(t)-K_{2}\bar{x}(t)\|^{2}]=0$ for $i=1,\ldots,N$.

Then, using Lemma \ref{difference-equation} and \eqref{x_in}, we have  
$$\lim_{t\to\infty}E[\big(x_{in}(t)-\bar{x}_{n}(t)\big)^{2}]=0.$$

The proof of Lemma \ref{kx-x} shows that $\mathds{D}x_{ij}(t)=x_{ij}(t+1)=x_{i(j+1)}(t)$ for $i=1,\ldots,N$ and $j=1,\ldots,n-1$.
Then, we have
\begin{align*}
&\lim_{t\to\infty}E[\big(x_{ij}(t)-\bar{x}_{j}(t)\big)^{2}]
\\
=&\lim_{t\to\infty}E[\big(x_{in}(t-n+j)-\bar{x}_{n}(t-n+j)\big)^{2}]
\\
=&0,
\end{align*} 
 where $i=1,\ldots,N$, $j=1,\ldots,n$.
Thus, $$\lim_{t\to\infty}E[\|x_{i}(t)-\bar{x}(t)\|^{2}]=0.$$

%

ii) Similarly, when $\beta>\frac{1}{2f_{M}}(\frac{\gamma^{2}\lambda_{G}^{2}}{\lambda_{2}^{2}(\gamma \lambda_{2}-1)}+\gamma\alpha+1)$ and $\gamma>\frac{1}{\lambda_{2}}$,  by Theorem \ref{compressed}, we have $$E[\|K_{2}x_{i}(t)-K_{2}\bar{x}(t)\|^{2}]=O(\frac{1}{t}).$$

Same as the proof of Part i), by Lemmas \ref{kx-x}-\ref{difference-equation}, there is $E[\big(x_{ij}(t)-\bar{x}_{j}(t)\big)^{2}]=O(\frac{1}{t})$, $i=1,\ldots, N$, $j=1,\ldots,n$. As a result, we know that 
$$E[\|x_{i}(t)-\bar{x}(t)\|^{2}]=O\left(\frac{1}{t}\right).
$$
This completes the proof.
\end{proof}

\section{Consensus of the MAS under Markovian switching communication networks}\label{switching}
In the previous sections, the effectiveness of Algorithm \ref{algorithm1} has been established for the fixed communication networks. 
Now, we will extend these results to the switching communication networks.

Model the communication links between agents as time-varying topology $G_{m(t)}=(N_0, E_{m(t)})$, 
whose dynamic is described by a homogeneous Markovian chain $\{m(t):t\in \mathds{N}\}$ with a state space $\{1,2,\ldots,h\}$, transition probability $p_{uv}=\mathbb{P}\{m(t)=v|m(t-1)=u\}$, and stationary distribution $\pi_{u}=\lim_{t\to\infty}\mathbb{P}\{m(t)=u\}$, for all $u,v \in \{1,2,\ldots, h\}$. 
$N_0= \{1, \ldots, N\}$ is the set of agents, and  $E_{m(t)}\subseteq N_{0}\times N_{0}$ is the ordered edges set of the topology $G_{m(t)}$. 
Moreover, assume $G_{m(t)}\in \{ G_{1},G_{2}, \ldots, G_{h} \}$. 
Denote $N_{i}^{m(t)}$  as the neighbor set of the agent $i$ in the topology $G_{m(t)}$.  
Denote the adjacency matrix and the degree matrix at time $t$ as $A_{m(t)}$ and $D_{m(t)}$, respectively.  
Then, the Laplace matrix of $G_{m(t)}$ is defined as $L_{m(t)} = D_{m(t)} - A_{m(t)}$.

To ensure the effectiveness of the algorithm, we give the following joint connectivity assumption.
\begin{assumption}\label{assm-Gi}
$\{G_{1},G_{2},\ldots, G_{h}\}$  are jointly connected.
\end{assumption}

Based on Assumption \ref{assm-Gi}, denote the jointly connected topology formed by $G_{1}$,$G_{2}$,$\ldots$,$ G_{h}$ as $G'=(N_0,E'),$ where  $E'=E_{1}\cup\cdots\cup E_{h}$ is the set of all the edges. 
Next, same as the fixed topology $G$ in Section \ref{fixed}, we define new notations for the new topology $G'$, which is considered in the following. 
Without loss of generality, we continue to use the same notations $d_{i}$, $N_{i}$, $d$, $\hat{z}(t)$, $s(t)$ and $C$ for $G'$ to simplify this paper.

Similarly, three matrices are constructed for the switching topology case.

$P_{m(t)}$ is designed to select each neighbor of each agent at time $t$.
Define $P_{m(t)}=\text{diag}\{p_{m(t)}^{11},p_{m(t)}^{22},\ldots, p_{m(t)}^{dd}\}\in \mathds{R}^{d \times d}$, 
where $p_{m(t)}^{ss}=1 $ when $(i,r_{s})\in E_{m(t)}$, else $p_{m(t)}^{ss}=0 $.

$W_{m(t)}$ is designed to select the neighbor set of each agent at time $t$.
Define $W_{m(t)}=[W_{m(t)}^{1},\ldots, W_{m(t)}^{N}]^{T}\in \mathds{R}^{N \times d},$ where
$W_{m(t)}^{i}=[\vec{0}_{d_{1}+\cdots+d_{i-1}},b_{m(t)}^{1},\ldots,b_{m(t)}^{d_{i}},\vec{0}_{d_{i+1}+\cdots+d_{N}}]^{T}\in \mathds{R}^{d}$ for  $i\in\{1,\ldots,N \}$ and
$k_{i}\in\{1,\ldots, d_{i}\}$,  $b_{m(t)}^{k_{i}}=1$ when $(i,r_{k_{i}+d_{1}+\cdots+d_{i-1}})\in E_{m(t)}$,  
else $b_{m(t)}^{k_{i}}=0$. 

$Q$ is defined as Section \ref{fixed} and is based on the topology $G'$.

Based on the above matrices, the vector forms of estimation and update are given as follows:

$1.$Estimation:
\begin{align}
\hat{z}(t)=&\bm{\Pi_{M}}\Big\{\hat{z}(t-1)+\frac{\beta}{t}P_{m(t)}\big(\mathcal{F}(C-\hat{z}(t-1))-s(t)\big)\Big\},
\label{Estimate}
\end{align}
where $\bm{\Pi_{M}}(z)=[\Pi_{M}(z_{1}),\ldots, \Pi_{M}(z_{d})]^{T}$, $\mathcal{F}(z)=[F(z_{1}),$
$\ldots,F(z_{d})]^{T}$, for any $z=[z_{1},z_{2},\ldots, z_{d}]^{T}\in \mathds{R}^{d}$.

$2.$Update:
\begin{align}\label{Update}
x(t)={}&\Big(I_{N}\otimes (A+BK_{1})-\frac{\gamma}{t}L_{m(t-1)}\otimes BK_{2}\Big)x(t-1)
\notag
\\
&+\frac{\gamma}{t}(W_{m(t-1)}\otimes B)\hat{\varepsilon}(t),
\end{align}
where  $\hat{\varepsilon}(t)=\hat{z}(t)-(Q\otimes K_{2})x(t)$ is the estimation error of the compressed state $(I_{N}\otimes K_{2})x(t)$.


The main difficulty in the analysis of the switching topology case is that the fixed matrices $I_{d}, L,$ and $W$ used to describe the relationship between agents are respectively changed into switching matrices $P_{m(t)}$, $L_{m(t)}$ and $W_{m(t)}$, in contrast to the fixed topology case.
To overcome this challenge, the following lemmas are given to establish a certain equivalence relationship between the Markovian switching topologies and a fixed topology.

\begin{lemma}\label{expectation}
For the switching matrices $L_{m(t)}$ and $P_{m(t)}$, we have the following:
$$E[L_{m(t)}]=\sum_{i=1}^{h}\pi_{i}L_{i}+O(\lambda_{L}^{t}),$$
$$E[P_{m(t)}]=\sum_{i=1}^{h}\pi_{i}P_{i}+O(\lambda_{P}^{t}),$$ 
where $0<\lambda_{L},\lambda_{P}<1$.
\end{lemma}
\begin{proof}
Denote $p_{u,t}=\mathbb{P}\{G_{m(t)}=G_{u}\}$, the transition probability matrix as $P=\{p_{uv}\}_{u,v}$ and stationary distribution vector as $\pi=[\pi_{1},\ldots,\pi_{h}]$.
By the property of transition probability $p_{uv}$ and stationary distribution $\pi_{u}$, we have $\pi P=\pi$,
which implies that $\pi$ is a positive left eigenvector corresponding to the eigenvalue $1$.

It can be seen that $P\vec{1}_{h}=\vec{1}_{h}$, thus $\vec{1}_{h}$ is a positive right eigenvector corresponding to the eigenvalue $1$. Besides, we can know that $\pi \vec{1}_{h}=\sum_{i=1}^{h}\pi_{i}=1$.

Since $P$ is symmetrical, by \cite[Corollary 1]{seneta-markov-2006}, we know that eigenvalue $1$ satisfies \cite[Theorem 1.1]{seneta-markov-2006}. Then, by \cite[Theorem 1.2]{seneta-markov-2006}, there exists a $\lambda_{p}$ that satisfies $0<\lambda_{p}<1$ such that 
$$P^{t}=\vec{1}_{h}\pi+O(\lambda_{p}^{t}).$$

Moreover, by the definition of $p_{u,t}$, it can be seen that $[p_{1,t+1},\ldots,p_{h,t+1}]=[p_{1,t},\ldots,p_{h,t}]P$.
Then, we have
\begin{align*}
[p_{1,t},\ldots,p_{h,t}]=&[p_{1,1},\ldots,p_{h,1}]P^{t-1}
\\
=&[p_{1,1},\ldots,p_{h,1}]\vec{1}_{h}\pi+O(\lambda_{p}^{t}),
\end{align*}
which implies that $p_{u,t}=\sum_{i=1}^{h}\pi_{u}p_{i,1}+O(\lambda_{p}^{t})=\pi_{u}+O(\lambda_{p}^{t})$, where $0<\lambda_{p}<1$.

Therefore, by the definition of expectation, we have
$E[L_{m(t)}]=\sum_{i=1}^{h}p_{i,t}L_{i}=\sum_{i=1}^{h}\pi_{i}L_{i}+O(\lambda_{L}^{t})$ and $E[P_{m(t)}]=\sum_{i=1}^{h}\pi_{i}P_{i}+O(\lambda_{P}^{t})$, where $0<\lambda_{L},\lambda_{P}<1$.
This completes the proof.
\end{proof}

Then, denote $\check{L}=\sum_{i=1}^{h}\pi_{i}L_{i}$, the following lemma shows its properties, similarly to Lemma \ref{L}. To simplify the notation, denote the eigenvalues of $\check{L}$ as $\lambda_{1}\le\lambda_{2}\le \lambda_{3}\le \ldots \le \lambda_{N}$, which is the same as Lemma \ref{L}.
 
\begin{lemma}\label{checkL}
(\!\!\cite{hu-consensus-SCIS2022}).
If Assumption \ref{assm-Gi} holds, then matrix $\check{L}$ has these properties:
\begin{enumerate}[i)]
\item $\check{L}$  is a nonnegative definite matrix with rank $n-1$ and eigenvalues $0=\lambda_{1}<\lambda_{2}\le \lambda_{3}\le \ldots \le \lambda_{N}$;
\item There exists an orthogonal matrix $T_{G}$ such that $T_{G}^{-1}\check{L}T_{G}=\text{diag}\{\lambda_{1},\ldots,\lambda_{N}\}$.
\end{enumerate}
\end{lemma}

\begin{remark}\label{E[Lt]=L}
The switching topology case can be converted into a fixed one by taking the expectation,
requiring only minor adjustments in the analysis process to achieve convergence properties consistent with the fixed topology case.
From Lemmas \ref{expectation} and \ref{checkL}, we can see that the switching matrix $L_{m(t)}$ can be treated as the sum of the fixed matrix $\check{L}$ and the exponential term $O(\lambda_{L}^{t})$ when taking a mathematical expectation.
Since $O(\lambda_{L}^{t})$ converges exponentially to zero, it is smaller than $O (\frac{1}{t^{2}})$, which does not affect the analysis of convergence.
\end{remark}

Based on these lemmas, the coupling relationship between the two Lyapunov functions $V(t)$ and $R(t)$ is obtained, which is in the same form as in the fixed graph case and only with a slight change in the definitions of the constant coefficients.

\begin{lemma}\label{Vi}
Under Assumptions \ref{assm-AB}-\ref{assm-d} and \ref{assm-Gi}, $V(t)$ satisfies
\begin{equation}
V(t)\le \Big(1-\frac{\gamma\lambda_{2}}{t}\Big)V(t-1)+\frac{\gamma\lambda_{G}}{t\lambda_{2}}R(t-1)+O\Big(\frac{1}{t^{2}}\Big),
\end{equation}
where $\lambda_{G}=\mathop{\max}\limits_{1\le i \le h}\{\|T_{G}^{-1}J_{N}W_{i}\|^{2}\}$. 
\end{lemma}
\begin{proof}
See Appendix \ref{app-Vi}.
\end{proof}

\begin{lemma}\label{Ri}
Under Assumptions \ref{assm-AB}-\ref{assm-d} and \ref{assm-Gi}, $R(t)$ satisfies
\begin{align}
R(t)\le  &\Big(1-\frac{2\beta f_{M}\pi_{\min}-\gamma\alpha}{t}\Big)R(t-1)+\frac{\gamma\lambda_{G}}{t\lambda_{2}}V(t-1)
\notag
\\
&+O\Big(\frac{1}{t^{2}}\Big),
\end{align}
where $f_{M}=\min_{i,j\in N_{0}}\{f(|c_{ij}|+M)\}$, $\pi_{\min}=\mathop{\min}\limits_{1\le i \le h}\{\pi_{i}\}$, $\alpha=2\sqrt{\lambda_{QW}}+\lambda_{QL}\lambda_{2}/\lambda_{G}$, $\lambda_{QW}=\mathop{\max}\limits_{1\le i \le h}\{\|QW_{i}\|^{2}\}$, and $\lambda_{QL}=\mathop{\max}\limits_{1\le i \le h}\{\|QL_{i}T_{G}\|^{2}\}$. 
\end{lemma}
\begin{proof}
See Appendix \ref{app-Ri}.
\end{proof}

Since the lemmas \ref{Vi}-\ref{Ri} have the same form as the lemmas \ref{V}-\ref{R}, we can just repeat the analysis process of the fixed topology case to get the results of the switching topology case, as the following theorems. 

\begin{theorem}\label{compressed-switching}
Under Assumptions \ref{assm-AB}-\ref{assm-d} and \ref{assm-Gi}, the compressed states $K_{2}x_{j}(t)$ and their estimates $\hat{z}_{ij}(t)$ satisfy:
\begin{enumerate}[i)]
\item When $\beta>\frac{1}{2f_{M}\pi_{\min}}(\frac{\gamma \lambda_{G}^{2}}{\lambda_{2}^{3}}+\gamma \alpha)$, the compressed states reach consensus and their estimates converge to the real compressed states, i.e., for $i=1,\ldots,N$, $j\in N_{i}$,
$$
\lim_{t\to\infty}E[\|K_{2}x_{i}(t)-K_{2}\bar{x}(t)\|^{2}]=0,
$$
$$
\lim_{t\to\infty}E[\|\hat{z}_{ij}(t)-K_{2}x_{j}(t)\|^{2}]=0;
$$
\item When $\beta>\frac{1}{2f_{M}\pi_{\min}}(\frac{\gamma^{2}\lambda_{G}^{2}}{\lambda_{2}^{2}(\gamma \lambda_{2}-1)}+\gamma\alpha+1)$ and $\gamma>\frac{1}{\lambda_{2}}$, the compressed states reach consensus at the rate of $O(\frac{1}{t})$, and their estimates converge to the real compressed states at the rate of $O(\frac{1}{t})$, i.e., for $i=1,\ldots, N$, $j\in N_{i}$,
$$
E[\|K_{2}x_{i}(t)-K_{2}\bar{x}(t)\|^{2}]=O\Big(\frac{1}{t}\Big),
$$
$$
E[\|\hat{z}_{ij}(t)-K_{2}x_{j}(t)\|^{2}]=O\Big(\frac{1}{t}\Big),
$$
\end{enumerate}
where $\bar{x}(t)=\frac{1}{N}\sum_{i=1}^{N}x_{i}(t)$.
\end{theorem}

\begin{theorem}\label{consensus-switching}
Under Assumptions \ref{assm-AB}-\ref{assm-d} and \ref{assm-K}-\ref{assm-Gi}, the following results are obtained:
\begin{enumerate}[i)]
\item The MAS \eqref{MAS}-\eqref{canonical} reaches consensus, i.e.,
$$
\lim_{t\to\infty}E[\|x_{i}(t)-\bar{x}(t)\|^{2}]=0,
$$
providing with $\beta>\frac{1}{2f_{M}\pi_{\min}}(\frac{\gamma \lambda_{G}^{2}}{\lambda_{2}^{3}}+\gamma \alpha)$;

\item The MAS \eqref{MAS}-\eqref{canonical} reaches consensus at the rate of $O(\frac{1}{t})$, i.e.,
$$
E[\|x_{i}(t)-\bar{x}(t)\|^{2}]=O\left(\frac{1}{t}\right),
$$
if $\beta>\frac{1}{2f_{M}\pi_{\min}}(\frac{\gamma^{2}\lambda_{G}^{2}}{\lambda_{2}^{2}(\gamma \lambda_{2}-1)}+\gamma\alpha+1)$ and $\gamma>\frac{1}{\lambda_{2}}$.
\end{enumerate}
\end{theorem}

\begin{remark}
The difference between the results of the Markovian switching topologies and the fixed topology is that the step coefficient $\beta$ depends on the switching probability of Markovian switching topologies, i.e., the minimal stationary distribution $\pi_{\min}$ of the Markovian chain $\{m(t):t\in\mathds{N}\}$.
In the case of a single topology, it is straightforward to observe that $\pi_{\min} = 1$. Consequently, Theorems \ref{compressed-switching}-\ref{consensus-switching} reduce to the results for the fixed topology case, as stated in Theorems \ref{compressed}-\ref{consensus}.
\end{remark}

\section{Numerical Simulation}\label{sec5}
In this section, we provide two simulation examples to illustrate the theoretical results. 

\begin{example}\label{example-fixed}
Consider the altitude consensus control of a multi-aircraft system composed of seven aircraft (\cite{wang-consensus-IJSS2015} and \cite{slotine-applied-1991}), whose communication network is shown as Figure \ref{topology}.

\begin{figure}[h]
\centering
\includegraphics[width = 0.2\textwidth]{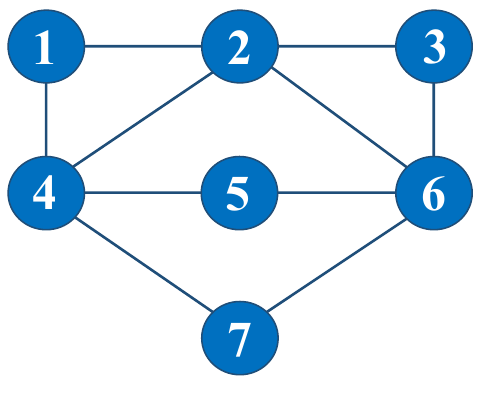}
\caption{The fixed communication topology of the multi-agent system}
\label{topology}
\end{figure}

The schematic diagram of the aircraft is shown in Figure \ref{aircraft}, where $L_{W}$ denotes the lift force applied at the center of lift $C_{L}$; $C_{G}$ is the center of mass; $d$ is the distance between $C_{L}$ and $C_{G}$. The mass of aircraft is denoted by $m$ and its moment of inertia about $C_{G}$ is denoted by $J$. The altitude of the $i$th aircraft is denoted by $h_{i}$, which is controlled by the elevator’s rotation $E_{i}$.

\begin{figure}[h]
\centering
\includegraphics[width = 0.4\textwidth]{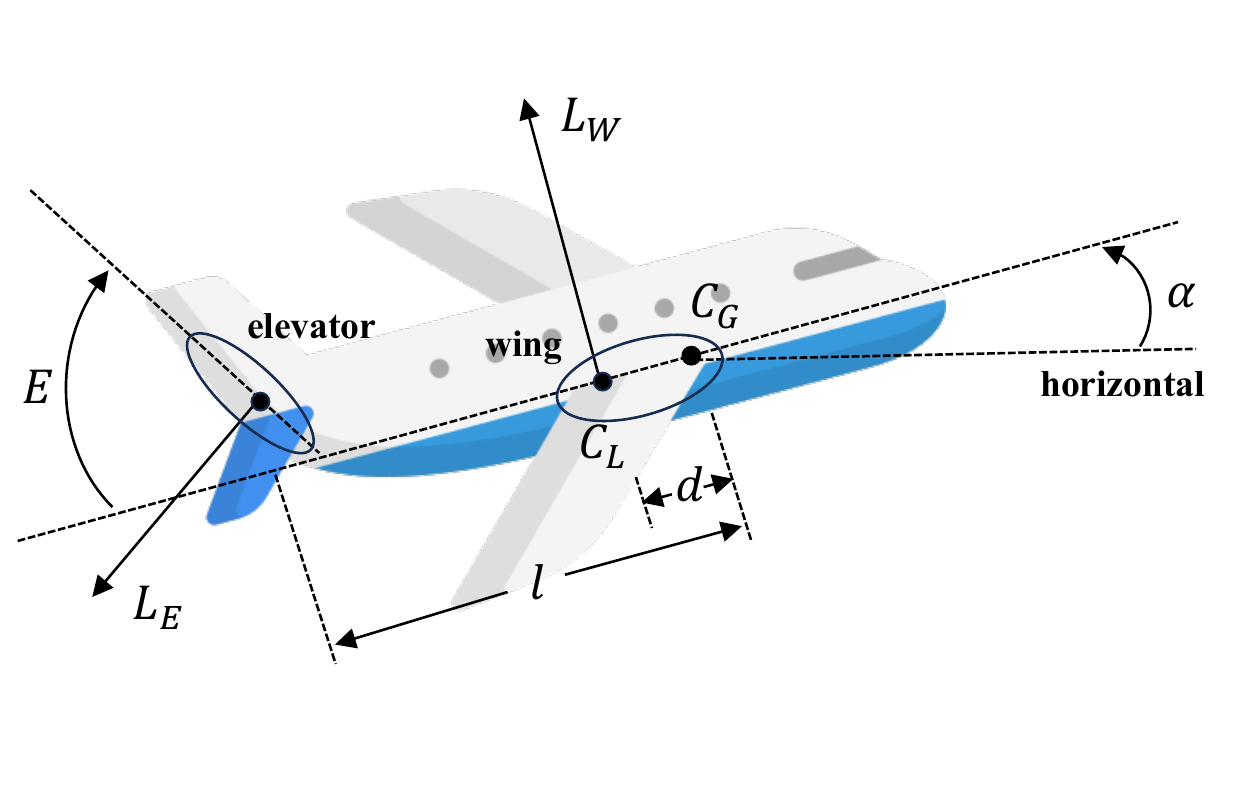}
\caption{A schematic diagram of an aircraft}
\label{aircraft}
\end{figure}

By \cite{slotine-applied-1991}, the dynamics of the $i$th aircraft can be modelled as follows:
\begin{equation}\label{dynamics}
\begin{cases}
J\ddot{\alpha}_{i}+b\dot{\alpha}_{i}+(C_{ZE}l+C_{ZW}d)\dot{\alpha}_{i}=C_{ZE}lE_{i},
\\
m\ddot{h}_{i}=(C_{ZE}+C_{ZW})\alpha_{i}-C_{ZE}E_{i},
\end{cases}
\end{equation}
where $\alpha_{i}$ is the rotation angle of the $i$th aircraft about $C_{G}$; $b$ is the friction coefficient; $L_{E}=C_{ZE}(E_{i}-\alpha_{i})$ is the aerodynamic force on the elevator.

In this numerical simulation, similar to \cite{wang-consensus-IJSS2015}, we set $J=1, m=1, b=4, C_{ZE}=1, C_{ZW}=5, l=3$ and $d=0.2$.
Let $x_{i}=(\alpha_{i},\dot{\alpha}_{i},h_{i},\dot{h}_{i})^{T}$. Then, the dynamics of the $i$th aircraft can be rewritten as
$$
\dot{x}_{i}=A_{c}x_{i}+B_{c}E_{i},
$$
where $$A_{c}=\begin{bmatrix}
0 & 1 & 0 & 0\\
-4&-4&  0& 0\\
0 & 0 & 0 &1\\
6 & 0 & 0&0
\end{bmatrix}, B_{c}=
\begin{bmatrix}
0\\
3\\
0\\
1
\end{bmatrix}.$$

Due to the wide application of digital networks, the data is sampled in practice. In the simulation, the sampling period is set to be $T=0.5$.
Then, by adopting the zero-order holder strategy, the discretization of \eqref{dynamics} is obtained as follows:
$$
x_{i}(t+1)=A_{d}x_{i}(t)+B_{d}E_{i}(t),
$$
where 
\begin{align}\label{AB}
A_{d}&=e^{A_{c}T}=\begin{bmatrix}
0.7358 & 0.1839 &0 &0\\
-0.7358 &0 &0&0\\
0.7073&0.0777&1&0.5\\
2.6891&0.3964&0&1
\end{bmatrix}, \notag
\\
B_{d}&=\int_{0}^{T}e^{A_{c}t}dtB_{c}=
\begin{bmatrix}
0.1982\\
0.5518\\
-0.093\\
-0.2668
\end{bmatrix}.
\end{align}
It is easily verified that the linear system described by \eqref{AB} is controllable. Therefore, it can be transformed into the Brunovsky canonical form defined by \eqref{canonical}.

In this simulation example, we set the communication noises $d_{ij}\sim N(0, 16)$ for $i\in N_{0}$, $j\in N_{i}$.
Then, assume that the initial altitudes $h_{i}$ of these seven aircraft are 
$h_{1}=5, h_{2}=2, h_{3}=4, h_{4}=3, h_{5}=1.5, h_{6}= 2.5, h_{7}=1$ and the initial values of $\alpha_{i},\dot{\alpha}_{i},\dot{h}_{i}$ are $0$.
The unit of $h_{i}$ can be selected as kilometers, meters, etc., according to the actual situation. Denote the initial estimates as $\vec{0}_{20}$.
Moreover, select $\beta=1500$ and $\gamma=1$. 
Then, we apply the consensus algorithm in Algorithm \ref{algorithm1}, with the thresholds $c_{ij}=-2$, the controller gain $K_{1}=[-0.9224, -0.1825, -0.0000, -0.1788]$ and the compression coefficient $K_{2}=[3.8734, 0.9054, 0.3575, 0.8772]$.
Given the initial states $x_{i}^{0}$ and compression coefficient $K_{2}$, choose $M=2$ as appropriate.

As shown in Figure \ref{height}, the altitudes of all agents reach consensus. 
Besides, Figure \ref{rate} shows a linear relationship between the logarithm of the mean square errors (MSEs) and the logarithm of the index $t$, which illustrates each agent can reach consensus at the rate of $O(\frac{1}{t})$. 

\begin{figure}[h]
\centering
\setcounter{figure}{2}
\includegraphics[width = 0.45\textwidth]{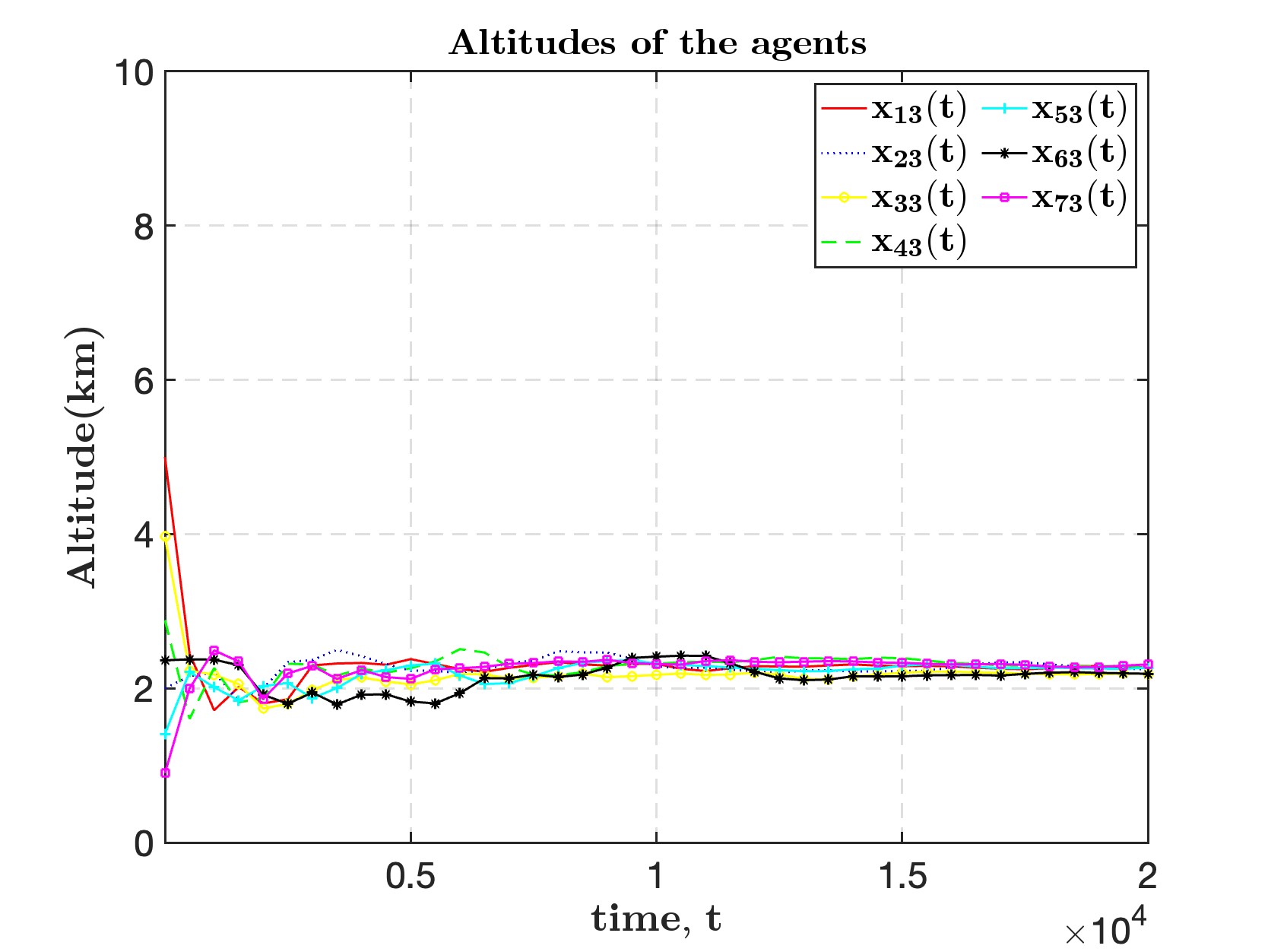}
\caption{The altitude trajectory of each aircraft under fixed topology}
\label{height}
\end{figure}

\begin{figure}[h]
\centering
\setcounter{figure}{3}
\includegraphics[width = 0.45\textwidth]{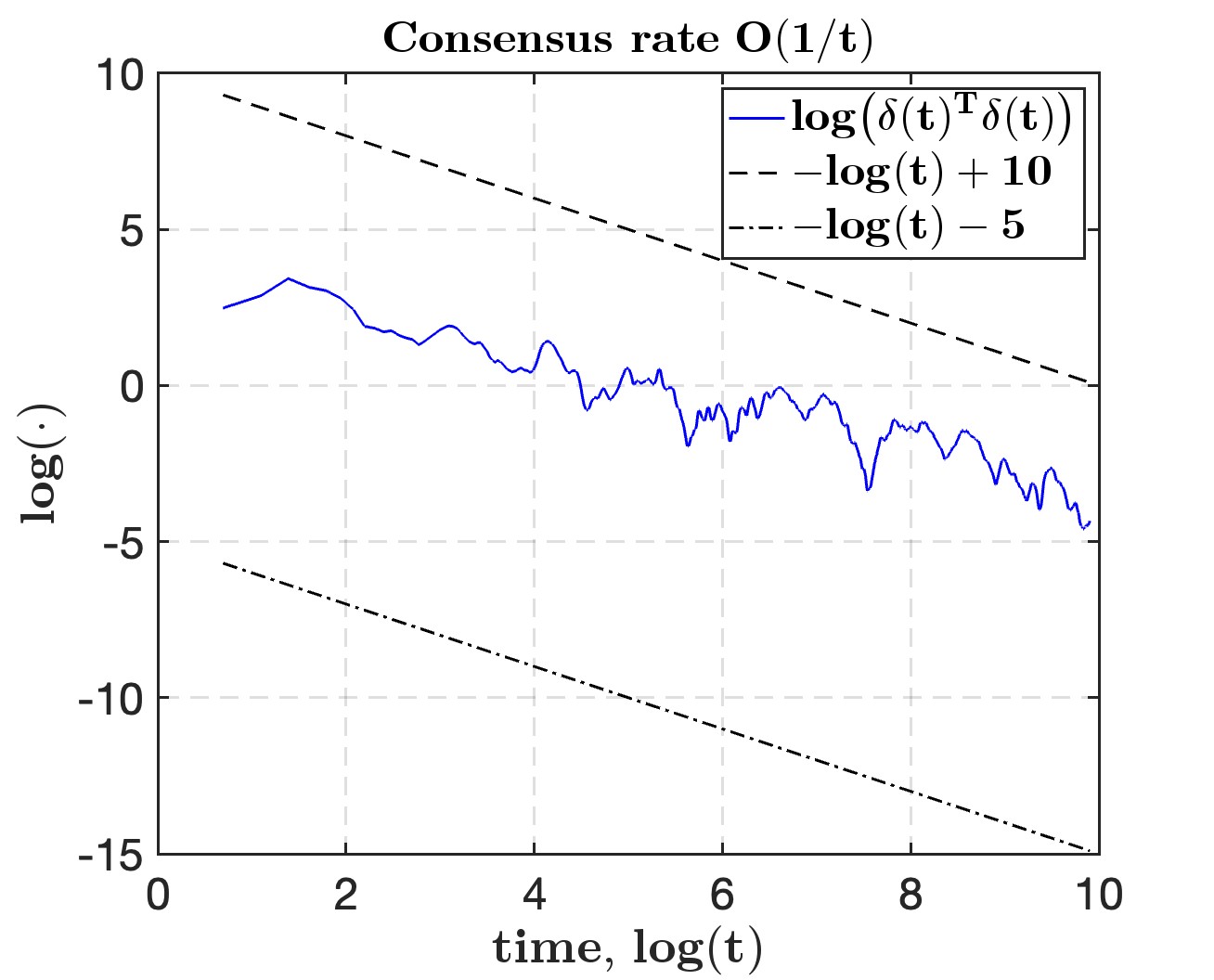}
\caption{The trajectory of the log(MSEs)  under fixed topology}
\label{rate}
\end{figure}

\end{example}

We next construct another simulation example to demonstrate the consensus of the MAS under Markovian switching communication networks.

\begin{figure}[h]
\centering
\includegraphics[width = 0.45\textwidth]{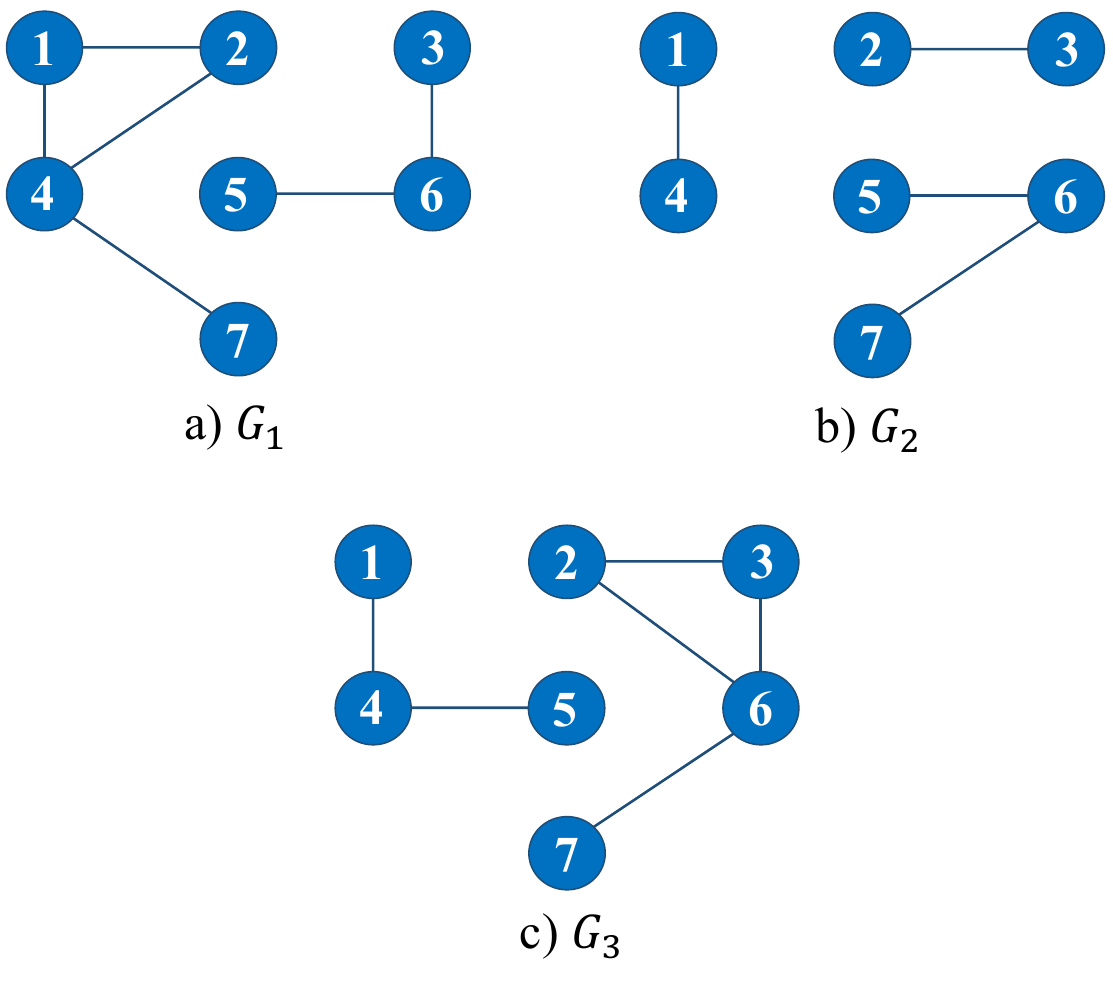}
\caption{The Markovian switching communication topologies}
\label{markov}
\end{figure}

\begin{example}\label{example-markov}
Take the same system settings, including the system model, noise distribution, and initial states, as those in Example \ref{example-fixed}.
The Markovian switching topologies $G_{m(t)}$ with $m(t)\in\{1,2,3\}$ are depicted in Figure \ref{markov}.
The transition probability matrix of the Markov chain $m(t)$ is chosen to be $P=
\begin{bmatrix}
0.5 & 0.3 & 0.2\\
0.2&0.5&0.3\\
0.3&0.2&0.5
\end{bmatrix}
$ with its stationary distribution is $\pi=[1/3, 1/3, 1/3]$.
It is clear that Assumption \ref{assm-Gi} is satisfied.

By Theorem \ref{consensus-switching}, we select $\beta=10000$, and $\gamma=2.4$.
Then, apply the consensus algorithm with the same thresholds $c_{ij}$, the controller gain $K_{1}$, and the compression coefficient $K_{2}$ as Example \ref{example-fixed}.
Figure \ref{markov-consensus} shows the consensus of the MAS under Markovian switching topologies.
Moreover, Figure \ref{markov-rate} illustrates the consensus rate $O(\frac{1}{t})$ of the MAS, which is the same as Example \ref{example-fixed}.

\begin{figure}[h]
\centering
\includegraphics[width = 0.45\textwidth]{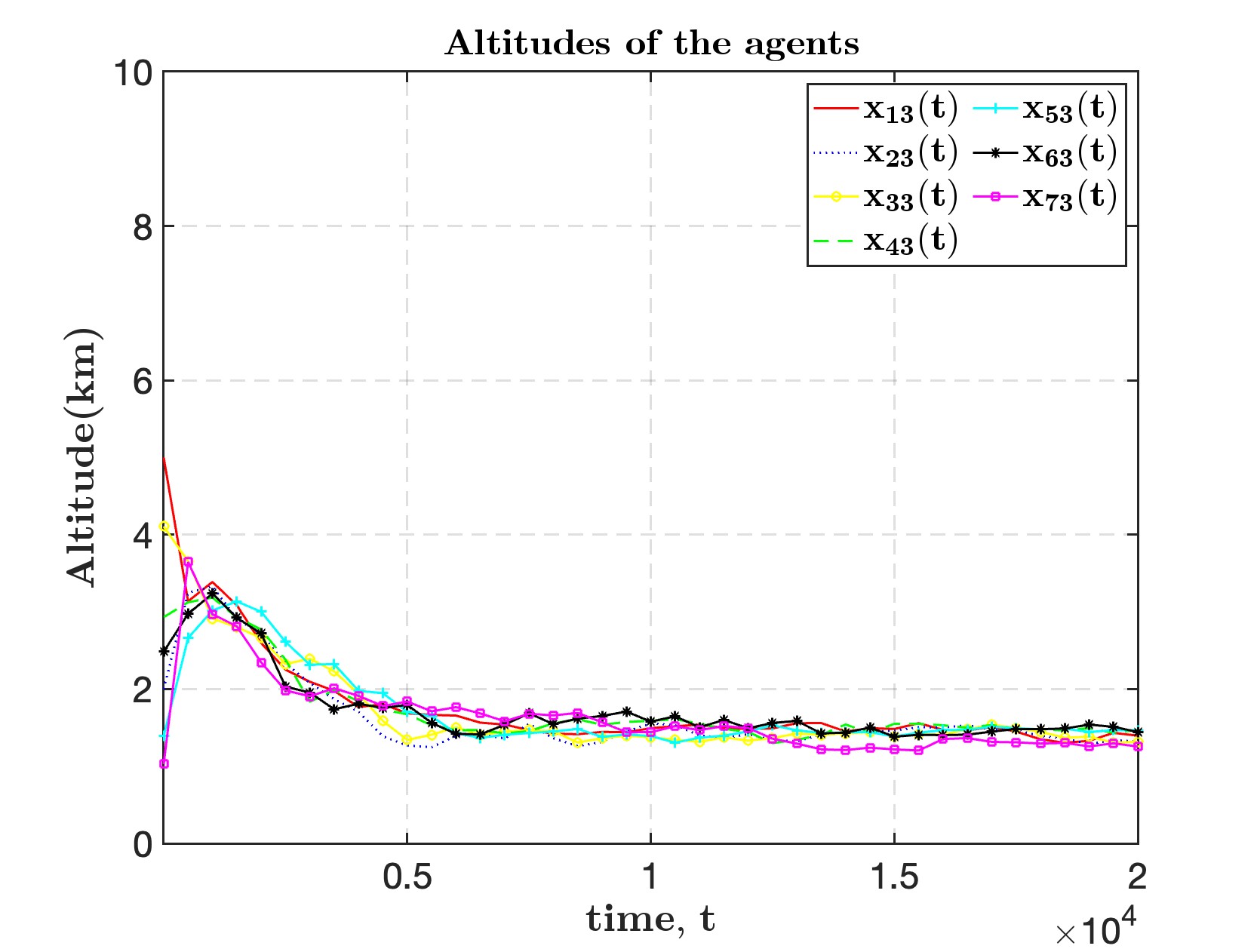}
\caption{The altitude trajectory of each aircraft under Markovian switching topologies}
\label{markov-consensus}
\end{figure}

\begin{figure}[h]
\centering
\includegraphics[width = 0.45\textwidth]{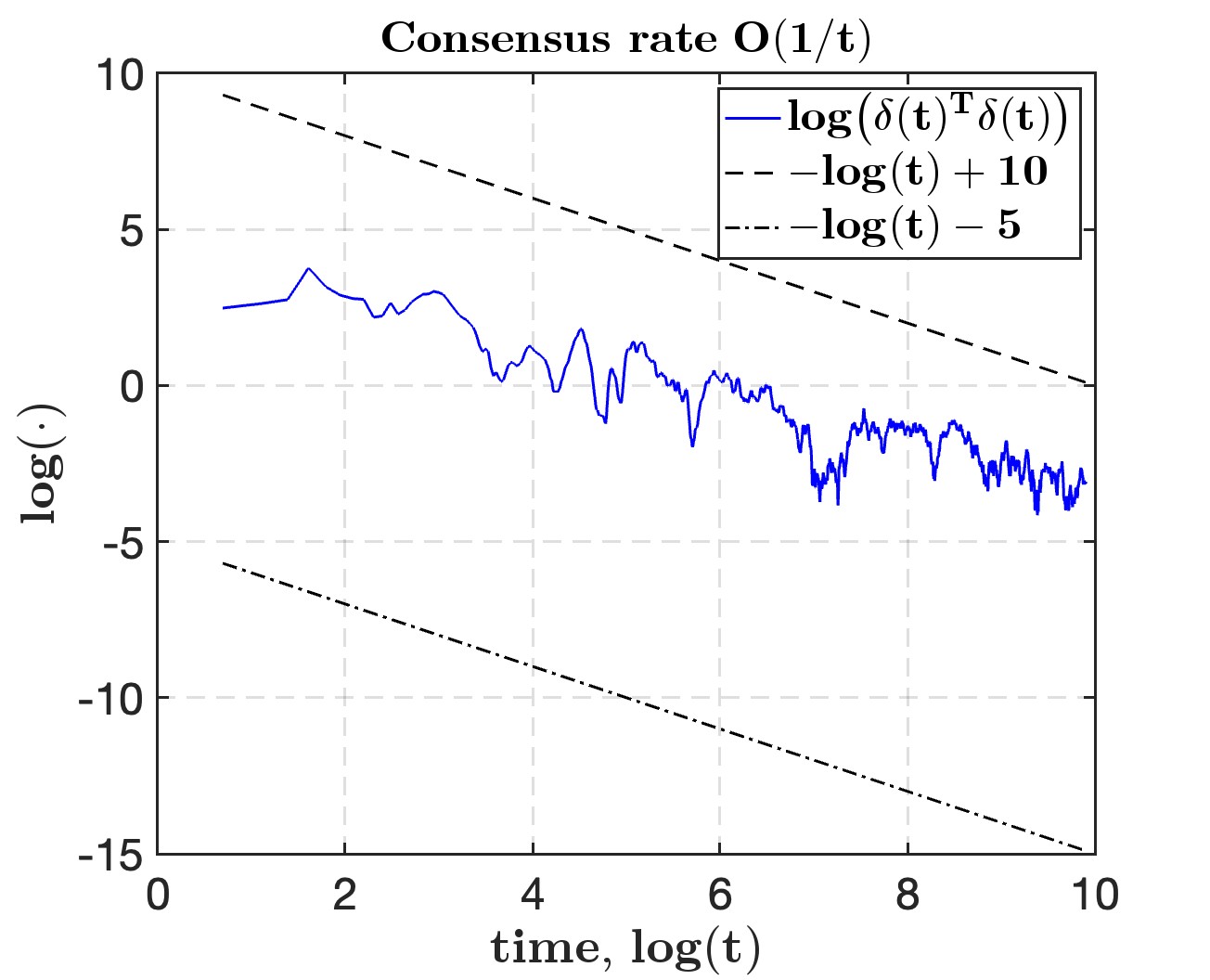}
\caption{The trajectory of the log(MSEs) under Markovian switching topologies}
\label{markov-rate}
\end{figure}

\begin{remark}
The connectivity and switching probability of communication topologies affect the selection of the step coefficients $\beta$ and $\gamma$. 
Specifically, the values of $\beta$ and $\gamma$ are affected by the algebraic connectivity $\lambda_{2}$ of the topology graph and minimum stationary distribution probability $\pi_{\min}$ related to the switching probability.
To achieve consensus, the smaller values of $\lambda_{2}$ and $\pi_{\min}$, the larger the step coefficients required, as shown in Examples \ref{example-fixed}-\ref{example-markov}.
To be specific, in the Markovian switching topology case, both the decrease of $\lambda_{2}$ and $\pi_{\min}$ increase the lower bound on the required estimation step coefficient $\beta$ and control step coefficient $\gamma$, necessitating a larger estimation step coefficient $\beta$.
This explains why the estimation step coefficient $\beta$ is selected as $\beta=1500$ in Example \ref{example-fixed} and $\beta=10000$ in Example \ref{example-markov}.
\end{remark}

\end{example}

\section{Conclusion}\label{sec-conclusion}
This paper investigates the one-bit consensus of controllable linear MASs with communication noises. 
A consensus algorithm consists of a communication protocol and a consensus controller is designed.
The communication protocol introduces a linear compression encoding function to achieve one-bit communication, which significantly saves communication costs.
A consensus controller with both a stabilization term and a consensus term is proposed to ensure consensus of an unstable MAS.
Two Lyapunov functions for the consensus error and estimation error of the compressed states are constructed.
By jointly analyzing the convergence property of them, it is shown that the compressed states can achieve consensus, and the estimates can converge to the real values.
Moreover, by establishing the consensus equivalence between the compressed and original states, it is proven that all agents can achieve consensus at the rate of the reciprocal of the iteration times, both in the case with fixed and switching communication networks.

In the future, there will be many interesting problems in one-bit consensus control.
For example, the one-bit consensus with event-triggered communication mechanisms has attracted wide attention, yet there are still many open questions and challenges in this area that warrant further investigation.

\appendices

\section{The proof about compressed states and estimates}
\subsection{The proof of Lemma \ref{V}}\label{app-V}
From $\delta(t)=(J_{N}\otimes I_{n})x(t)$, $L\vec{1}_{N}=\vec{1}_{N}^{T}L=0$, and $J_{N}L=LJ_{N}$, we have
\begin{align*}\label{delta(t)}
\delta(t)=&(J_{N}\otimes I_{n})x(t)
\\
=&(J_{N}\otimes I_{n})(I_{n}\otimes \tilde{A}-\frac{\gamma}{t}L\otimes BK_{2})x(t-1)
\\
&+\frac{\gamma}{t}(J_{N}W\otimes B)\hat{\varepsilon}(t-1)
\\
=&(I_{N}\otimes \tilde{A}-\frac{\gamma}{t}L\otimes BK_{2})\delta(t-1)
\\
&+\frac{\gamma}{t}(J_{N}W\otimes B)\hat{\varepsilon}(t-1),
\tag{A1}
\end{align*}
where $\tilde{A}=A+BK_{1}$, $\delta(t)=[\delta_{1}^{T}(t),\ldots,\delta_{N}^{T}(t)]^{T}\in\mathds{R}^{nN}$, and $\delta_{i}(t)=x_{i}(t)-\bar{x}(t)$ for $i=1,\ldots,N$.

Denote $\tilde{\delta}(t)=(T_{G}^{-1}\otimes I_{n})\delta(t)$. 
Denote the first $n$ elements of $\tilde{\delta}(t)$ by $\tilde{\delta}^{(1)}(t)$, and the others by $\tilde{\delta}^{(2)}(t)$. Since the first row of $T_{G}^{-1}$ is $\frac{1}{\sqrt{N}}\vec{1}_{N}^{T}$, we know $\tilde{\delta}^{(1)}(t)=\frac{1}{\sqrt{N}}\sum_{i=1}^{N}\delta_{i}(t)=\vec{0}_{n}$.

Denote $\hat{\delta}(t)=(I_{N}\otimes K_{2})\tilde{\delta}(t)$.
Denote the first element of $\hat{\delta}(t)$ by $\hat{\delta}^{(1)}(t)$, and the others by $\hat{\delta}^{(2)}(t)$.
Since $\hat{\delta}^{(1)}(t)=K_{2}\tilde{\delta}^{(1)}(t)$ and $\tilde{\delta}^{(1)}(t)=\vec{0}_{n}$, we know $\hat{\delta}^{(1)}(t)=0.$
Besides, it can be seen that $V(t)=E[\|\hat{\delta}(t)\|^{2}]=E[\hat{\delta}^{T}(t)\hat{\delta}(t)]=E[\hat{\delta}^{(2)T}(t)\hat{\delta}^{(2)}(t)]$.

Then, by Remark \ref{K1K2} and \eqref{delta(t)}, we have
\begin{align*}\label{V-v1v2}
V(t)=&E[\|(T_{G}^{-1}\otimes K_{2})\delta(t)\|^{2}]
\\
=&E\big[\big\|(T_{G}^{-1}\otimes K_{2})\big((I_{N}\otimes \tilde{A}-\frac{\gamma}{t}L\otimes BK_{2})\delta(t-1)
\\
&+\frac{\gamma}{t}(J_{N}W\otimes B)\hat{\varepsilon}(t-1)\big)\big\|^{2}\big]
\\
=&E\big[\big\|(I_{N}\otimes K_{2}-\frac{\gamma}{t}T_{G}^{-1}LT_{G}\otimes K_{2})\tilde{\delta}(t-1)
\\
&+\frac{\gamma}{t}(T_{G}^{-1}J_{N}W)\hat{\varepsilon}(t-1)\big\|^{2}\big]
\\
=&E\big[\big\|(I_{N}-\frac{\gamma}{t}T_{G}^{-1}LT_{G})\hat{\delta}(t-1)
\\
&+\frac{\gamma}{t}(T_{G}^{-1}J_{N}W)\hat{\varepsilon}(t-1)\big\|^{2}\big]
\\
=&E[\hat{\delta}^{T}(t-1)(I_{N}-\frac{\gamma}{t}T_{G}^{-1}LT_{G})^{2}\hat{\delta}(t-1)]
\\
&+\frac{2\gamma}{t}E[\hat{\delta}^{T}(t-1)(I_{N}-\frac{\gamma}{t}T_{G}^{-1}LT_{G})
\\
&\cdot(T_{G}^{-1}J_{N}W)\hat{\varepsilon}(t-1)]+O\Big(\frac{1}{t^{2}}\Big).
\tag{A2}
\end{align*}
Denote the first and second items of formula \eqref{V-v1v2} as $V_{1}(t)$ and $V_{2}(t)$, respectively, i.e., $V(t)\triangleq V_{1}(t)+V_{2}(t)+O\big(\frac{1}{t^{2}}\big)$. 

By $\hat{\delta}^{(1)}(t)=0$, we have $\hat{\delta}^{T}(t)(I_{N}-\frac{\gamma}{t}\text{diag}(0,\lambda_{2},\ldots,\lambda_{N}))$
$\cdot\hat{\delta}(t)=\hat{\delta}^{(2)T}(t)(I_{N-1}-\frac{\gamma}{t}\text{diag}(\lambda_{2},\ldots,\lambda_{N}))\hat{\delta}^{(2)}(t)$.
Therefore, by Lemma \ref{L}, one can get
\begin{align*}\label{v1}
V_{1}(t)=&E[\hat{\delta}^{T}(t-1)(I_{N}-\frac{\gamma}{t}T_{G}^{-1}LT_{G})^{2}\hat{\delta}(t-1)]
\\
=&E\Big[\hat{\delta}^{T}(t-1)\text{diag}^{2}(1,1-\frac{\gamma}{t}\lambda_{2}\ldots,1-\frac{\gamma}{t}\lambda_{N})\hat{\delta}^{}(t-1)\Big]
\\
=&E\Big[\hat{\delta}^{(2)T}(t-1)\text{diag}^{2}(1-\frac{\gamma}{t}\lambda_{2}\ldots,1-\frac{\gamma}{t}\lambda_{N})\hat{\delta}^{(2)}(t-1)\Big]
\\
\le &\big(1-\frac{\gamma\lambda_{2}}{t}\big)^{2}V(t-1).
\tag{A3}
\end{align*}

Using the Cauchy-Schwarz inequality, we have
\begin{align*}\label{v2}
V_{2}(t)=&\frac{2\gamma}{t}E[\hat{\delta}^{T}(t-1)(I_{N}-\frac{\gamma}{t}T_{G}^{-1}LT_{G})
\\
&\cdot(T_{G}^{-1}J_{N}W)\hat{\varepsilon}(t-1)]
\\
=&\frac{2\gamma}{t}E[\hat{\delta}^{T}(t-1)\text{diag}^{}(1,1-\frac{\gamma}{t}\lambda_{2}\ldots,1-\frac{\gamma}{t}\lambda_{N})
\\
&\cdot(T_{G}^{-1}J_{N}W)\hat{\varepsilon}(t-1)]
\\
\le &\frac{2\gamma}{t} \Big(E\big[\hat{\delta}^{T}(t-1)\text{diag}^{2}(1,1-\frac{\gamma}{t}\lambda_{2},\ldots,1-\frac{\gamma}{t}\lambda_{N})
\\
&\cdot\hat{\delta}^{}(t-1)\big]E\big[\hat{\varepsilon}^{T}(t-1)W^{T}J_{N}^{T}J_{N}W\hat{\varepsilon}(t-1)\big]\Big)^{\frac{1}{2}}
\\
= &\frac{2\gamma}{t} \Big(E\big[\hat{\delta}^{(2)T}(t-1)\text{diag}^{2}(1-\frac{\gamma}{t}\lambda_{2},\ldots,1-\frac{\gamma}{t}\lambda_{N})
\\
&\cdot\hat{\delta}^{(2)}(t-1)\big]E\big[\hat{\varepsilon}^{T}(t-1)W^{T}J_{N}^{T}J_{N}W\hat{\varepsilon}(t-1)\big]\Big)^{\frac{1}{2}}
\\
\le & \frac{2\gamma}{t}\Big(\big(1-\frac{\gamma\lambda_{2}}{t}\big)^{2}V(t-1)\Big)^{\frac{1}{2}}\Big(\lambda_{G}R(t-1)\Big)^{\frac{1}{2}}
\\
\le & \frac{\gamma}{t}\Big(\lambda_{2}\big(1-\frac{\gamma\lambda_{2}}{t}\big)^{2}V(t-1)+\frac{\lambda_{G}}{\lambda_{2}}R(t-1)\Big).
\tag{A4}
\end{align*}

Combining \eqref{V-v1v2}-\eqref{v2}, we have
\begin{align*}
V(t)\le& \big(1+\frac{\gamma\lambda_{2}}{t}\big)\big(1-\frac{\gamma\lambda_{2}}{t}\big)^{2}V(t-1)+\frac{\gamma\lambda_{G}}{\lambda_{2} t}R(t-1)
\\
&+O\Big(\frac{1}{t^{2}}\Big)
\\
=& \big(1-\frac{\gamma\lambda_{2}}{t}-\frac{\gamma^{2}\lambda_{2}^{2}}{t^{2}}+\frac{\gamma^{3}\lambda_{2}^{3}}{t^{3}}\big)V(t-1)+\frac{\gamma\lambda_{G}}{\lambda_{2} t}R(t-1)
\\
&+O\Big(\frac{1}{t^{2}}\Big)
\\
\le & \big(1-\frac{\gamma\lambda_{2}}{t}\big)V(t-1)+\frac{\gamma\lambda_{G}/\lambda_{2}}{t}R(t-1)+O\Big(\frac{1}{t^{2}}\Big).
\end{align*}

\subsection{The proof of Lemma \ref{R}}\label{app-R}

By the definition of $\hat{\varepsilon}(t)$, Lemma \ref{bounded} and Remark \ref{projection}, we have 
\begin{align*}\label{R-r1r2r3}
R(t)=& E[\|\hat{\varepsilon}(t)\|^{2}]
\\
=& E[\|\hat{z}(t)-(Q\otimes K_{2})x(t)\|^{2}]
\\
=& E\Big[\Big\|\bm{\Pi_{M}}\Big\{\hat{z}(t-1)+\frac{\beta}{t}\big(\mathcal{F}(C-\hat{z}(t-1))-s(t)\big)\Big\}
\\
&-(Q\otimes K_{2})x(t)\Big\|^{2}\Big]
\\
\le & E\Big[\Big\|\hat{z}(t-1)+\frac{\beta}{t}\big(\mathcal{F}(C-\hat{z}(t-1))-s(t)\big)
\\
&-(Q\otimes K_{2})x(t)\Big\|^{2}\Big]
\\
=& E\Big[\Big\|\big(I_{d}-\frac{\gamma}{t}QW\big)\hat{\varepsilon}(t-1)+\frac{\gamma}{t}(QL\otimes K_{2})\delta(t-1)
\\
&+\frac{\beta}{t}\big(\mathcal{F}(C-\hat{z}(t-1))-s(t)\big)\Big\|^{2}\Big]
\\
=& E\Big[\Big\|\big(I_{d}-\frac{\gamma}{t}QW\big)\hat{\varepsilon}(t-1)+\frac{\gamma}{t}QLT_{G}\hat{\delta}(t-1)
\\
&+\frac{\beta}{t}\big(\mathcal{F}(C-\hat{z}(t-1))-s(t)\big)\Big\|^{2}\Big]
\\
=&E\Big[\hat{\varepsilon}^{T}(t-1)\big(I_{d}-\frac{\gamma}{t}QW\big)^{T}\big(I_{d}-\frac{\gamma}{t}QW\big)\hat{\varepsilon}(t-1)\Big]
\\
&+\frac{2\gamma}{t}E\Big[\hat{\varepsilon}^{T}(t-1)\big(I_{d}-\frac{\gamma}{t}QW\big)^{T}QLT_{G}\hat{\delta}(t-1)\Big]
\\
&+\frac{2\beta}{t}E\Big[\hat{\varepsilon}^{T}(t-1)\big(I_{d}-\frac{\gamma}{t}QW\big)^{T}\Big(\mathcal{F}\big(C-\hat{z}(t-1)\big)
\\
&-s(t)\big)\Big)\Big]+O\Big(\frac{1}{t^{2}}\Big).
\tag{A5}
\end{align*}
 Denote the first, second, and third items of \eqref{R-r1r2r3} as $R_{1}(t)$, $R_{2}(t)$, and $R_{3}(t)$, respectively, i.e., $R(t)\le R_{1}(t)+R_{2}(t)+R_{3}(t)+O\big(\frac{1}{t^{2}}\big)$. Then, we have
\begin{align*}\label{r1}
R_{1}(t)=&E\Big[\hat{\varepsilon}^{T}(t-1)\big(I_{d}-\frac{\gamma}{t}QW\big)^{T}\big(I_{d}-\frac{\gamma}{t}QW\big)\hat{\varepsilon}(t-1)\Big]
\\
&\le \big(1+\frac{\gamma\sqrt{\lambda_{QW}}}{t}\big)^{2}R(t-1).
\tag{A6}
\end{align*}

Similarly to \eqref{v2}, using the Cauchy-Schwarz inequality, gives
\begin{align*}\label{r2}
R_{2}(t)=&\frac{2\gamma}{t}E\Big[\hat{\varepsilon}^{T}(t-1)\big(I_{d}-\frac{\gamma}{t}QW\big)^{T}QLT_{G}\hat{\delta}(t-1)\Big]
\\
\le &\frac{2\gamma}{t}\Big(E\Big[\hat{\varepsilon}^{T}(t-1)\big(I_{d}-\frac{\gamma}{t}QW\big)^{T}\big(I_{d}-\frac{\gamma}{t}QW\big)
\\
&\cdot \hat{\varepsilon}(t-1)\Big]E[\hat{\delta}^{T}(t-1)T_{G}^{T}L^{T}Q^{T}QLT_{G}\hat{\delta}(t-1)]\Big)^{\frac{1}{2}}
\\
\le& \frac{2\gamma}{t}\Big(\big(1+\frac{\gamma\sqrt{\lambda_{QW}}}{t}\big)^{2}R(t-1)\cdot\lambda_{QL}V(t-1)\Big)^{\frac{1}{2}}
\\
\le & \frac{\gamma}{t}\Big(\frac{\lambda_{QL}\lambda_{2}}{\lambda_{G}}\big(1+\frac{\gamma\sqrt{\lambda_{QW}}}{t}\big)^{2}R(t-1)+\frac{\lambda_{G}}{\lambda_{2}}V(t-1)\Big).
\tag{A7}
\end{align*}

Moreover, by the communication protocol designed in Algorithm \ref{algorithm1}, it can be seen that $E[s(t)]=\mathcal{F}\big(C-(Q\otimes K_{2})x(t)\big)$ under Assumption \ref{assm-d}. Then, we get
\begin{align*}
R_{3}(t)=&\frac{2\beta}{t}E\Big[\hat{\varepsilon}^{T}(t-1)\big(I_{d}-\frac{\gamma}{t}QW\big)^{T}\Big(\mathcal{F}\big(C-\hat{z}(t-1)\big)
\\
&-s(t)\big)\Big]
\\
=&\frac{2\beta}{t}E\Big[\hat{\varepsilon}^{T}(t-1)\big(I_{d}-\frac{\gamma}{t}QW\big)^{T}\Big(\mathcal{F}\big(C-\hat{z}(t-1)\big)
\\
&-\mathcal{F}\big(C-(Q\otimes K_{2})x(t)\big)\Big)\Big]
\end{align*}
And, by Lagrange's Mean Value Theorem, we have 
\begin{align*}
&F\big(c_{ij}-\hat{z}_{ij}(t-1)\big)-F\big(c_{ij}-K_{2}x_{j}(t)\big)
\\
=&-f(\zeta_{ij}(t))\big(\hat{z}_{ij}(t-1)-K_{2}x_{j}(t)\big),
\end{align*}
 where $\zeta_{ij}(t)$ is between $c_{ij}-\hat{z}_{ij}(t-1)$ and $c_{ij}-K_{2}x_j(t).$

Let $\zeta(t)=[\zeta_{1r_{1}}(t),\ldots,\zeta_{ir_{s}}(t),\ldots,\zeta_{Nr_{d_{1}+\ldots+d_{N}}}(t)]^T,$ with $r_{s}$ representing the neighbor $j$ of agent $i$, i.e., $\zeta_{ir_{s}}(t)=\zeta_{ij}(t)$. Denote $\text{diag} (\vec{f}(\zeta(t)))=\text{diag}\{f(\zeta_{1r_{1}}(t)),\ldots,f(\zeta_{Nr_{d}}(t))\}\in\mathds{R}^{d\times d}$ as a diagonal matrix generated by each element of the vector $\vec{f}(\zeta(t))=[f(\zeta_{1r_{1}}(t),\ldots,f(\zeta_{Nr_{d}}(t))]^{T}\in \mathds{R}^{d}$. By Lemma \ref{bounded}, $\zeta_{ij}(t)$ is bounded. Since the function $f(\cdot)$ is  continuous, we have $\text{diag}\big(\vec{f}(\zeta(t))\big)\ge f_{M}\cdot I_{d}$ and 

\begin{align*}\label{r3}
R_{3}(t)=&-\frac{2\beta}{t}E\Big[\hat{\varepsilon}^{T}(t-1)\big(I_{d}-\frac{\gamma}{t}QW\big)^{T}
\text{diag}\big(\vec{f}(\zeta(t))\big)
\\
&\cdot\big(\hat{z}(t-1)-(Q\otimes K_{2})x(t)\big)
\\
=& -\frac{2\beta}{t}E\Big[\hat{\varepsilon}^{T}(t-1)\big(I_{d}-\frac{\gamma}{t}QW\big)^{T}\text{diag}\big(\vec{f}(\zeta(t))\big)
\\
&\cdot \Big(\hat{\varepsilon}(t-1)-\frac{\gamma}{t}\big(QW\hat{\varepsilon}(t-1)-QLT_{G}\hat{\delta}(t-1)\big)\Big)\Big]
\\
=& -\frac{2\beta}{t}E\big[\hat{\varepsilon}^{T}(t-1)\text{diag}\big(\vec{f}(\zeta(t))\big)\hat{\varepsilon}(t-1)\big]+O\Big(\frac{1}{t^{2}}\Big)
\\
\le &-\frac{2\beta f_{M}}{t}R(t-1)+O\Big(\frac{1}{t^{2}}\Big).
\tag{A8}
\end{align*}

Considering \eqref{R-r1r2r3} with \eqref{r1}-\eqref{r3}, we can obtain that
\begin{align*}
R(t)\le  &\Big(1-\frac{2\beta f_{M}-\gamma\alpha}{t}\Big)R(t-1)+\frac{\gamma\lambda_{G}/\lambda_{2}}{t}V(t-1)
\\
&+O\Big(\frac{1}{t^{2}}\Big).
\end{align*}

\section{The proof about original states}

\subsection{The proof of Lemma \ref{kx-x}}\label{app-kx}
By the Brunovsky canonical form \eqref{canonical} and \eqref{update}, we have
\begin{align*}
x_{i}(t+1)=&(A+BK_{1})x_{i}(t)+\frac{\gamma}{t+1}B\sum_{j\in N_{i}}(\hat{z}_{ij}(t)-K_{2}x_{j}(t))
\\
=&\begin{bmatrix}
0         &1         &\cdots        &0\\
\vdots & \vdots &\ddots &\vdots\\
0         & 0         & \cdots & 1      \\
b_{1} & b_{2}-b_{1}&\cdots&1-b_{n-1}
\end{bmatrix}x_{i}(t)
\\
+&\frac{\gamma}{t+1}
\begin{bmatrix}
0\\
\vdots\\
0\\
1
\end{bmatrix}
\sum_{j\in N_{i}}(\hat{z}_{ij}(t)-K_{2}x_{j}(t)).
\end{align*}
Then, there is $x_{i1}(t+1)=x_{i2}(t),\ldots,x_{i(n-1)}(t+1)=x_{in}(t)$, i.e.,
$$
\mathds{D}x_{ij}(t)=x_{i(j+1)}(t),\quad j=1,\ldots,n-1.
$$

Since $K_{2}x_{i}(t)=b_{1}x_{i1}(t)+\ldots,+b_{n-1}x_{i(n-1)}(t)+x_{in}(t)$, we can obatin that
$$
\mathds{D}^{n-1}K_{2}x_{i}(t)=b_{1}x_{in}(t)+\ldots+b_{n-1}\mathds{D}^{n-2}x_{in}(t)+\mathds{D}^{n-1}x_{in}(t).
$$

\subsection{The proof of Lemma \ref{difference-equation}}\label{app-de}
i) Let $\xi_{1}(t)\triangleq\prod_{i=2}^{n-1}(\mathds{D}-r_{i})\xi(t)$. Then, $\mathds{D}\xi_{1}(t)=r_{1}\xi(t)+\eta(t)$, i.e., $\xi_{1}(t+1)=r_{1}\xi_{1}(t)+\eta(t)$, and thus,
\begin{align*}\label{xi1}
&\xi_{1}(t)=r_{1}^{t}\xi_{1}(0)+\sum_{i=0}^{t-1}r_{1}^{i}\eta(t-1-i)
\\
=& r_{1}^{t}\xi_{1}(0)+\sum_{i=0}^{t-1}r_{1}^{i}\big(\eta(t-1-i)-\eta^{*}\big)+\sum_{i=0}^{t-1}r_{1}^{i}\eta^{*},
\tag{B1}
\end{align*}
where
\begin{align*}\label{eta-eta*}
&E\Big[\Big(\sum_{i=0}^{t-1}r_{1}^{i}\big(\eta(t-1-i)-\eta^{*}\big)\Big)^{2}\Big]
\\
=&E\Big[ \sum_{i=0}^{t-1}r_{1}^{2i}\big(\eta(t-1-i)-\eta^{*}\big)^{2}+\sum_{i=0}^{t-1}\sum_{j\neq i}r_{1}^{i+j}\big(\eta(t-1-i)
\\
&-\eta^{*}\big)\big(\eta(t-1-j)-\eta^{*}\big)\Big]
\\
=&\sum_{i=0}^{t-1}r_{1}^{2i}E\big[\big(\eta(t-1-i)-\eta^{*}\big)^{2}\big]+\sum_{i=0}^{t-1}\sum_{j\neq i}r_{1}^{i+j}E\big[\big(\eta(t-1
\\
&-i)-\eta^{*}\big)\big(\eta(t-1-j)-\eta^{*}\big)\big]
\\
\le &\sum_{i=0}^{t-1}r_{1}^{2i}E\big[\big(\eta(t-1-i)-\eta^{*}\big)^{2}\big]+\sum_{i=0}^{t-1}\sum_{j\neq i}r_{1}^{i+j}\Big(E\big[\big(\eta(t-1
\\
&-i)-\eta^{*}\big)^{2}\big]E\big[\big(\eta(t-1-j)-\eta^{*}\big)^{2}\big]\Big)^{\frac{1}{2}}
\\
=&\Big(\sum_{i=0}^{t-1}r_{1}^{i}\Big(E\big[\big(\eta(t-1-i)-\eta^{*}\big)^{2}\big]\Big)^{\frac{1}{2}}\Big)^{2}.
\tag{B2}
\end{align*}

To calculate the above formula \eqref{eta-eta*}, without loss of generality, assume that $r_{1}\ge0$. Then,
\begin{align*}\label{t-i}
&\lim_{t\to\infty}\Big|\sum_{i=0}^{t-1}r_{1}^{i}\Big(E\big[\big(\eta(t-1-i)-\eta^{*}\big)^{2}\big]\Big)^{\frac{1}{2}}\Big|
\\
=&\lim_{t\to\infty}\frac{\sum_{i=0}^{t-1}r_{1}^{-i}\Big(E\big[\big(\eta(i)-\eta^{*}\big)^{2}\big]\Big)^{\frac{1}{2}}}{r_{1}^{1-t}}.
\tag{B3}
\end{align*}

Since $0\le r_{1}<1$, we know that  $r_{1}^{1-t}$ is a strictly monotone and divergent sequence. 
Then, by Stolz-Ces\`aro theorem,
\begin{align*}
&\lim_{t\to\infty}\frac{\sum_{i=0}^{t-1}r_{1}^{-i}\Big(E\big[\big(\eta(i)-\eta^{*}\big)^{2}\big]\Big)^{\frac{1}{2}}}{r_{1}^{1-t}}
\\
=&\lim_{t\to\infty}\frac{r_{1}^{-t}\Big(E\big[\big(\eta(t)-\eta^{*}\big)^{2}\big]\Big)^{\frac{1}{2}}}{r_{1}^{-t}-r_{1}^{1-t}}
\\
=& \lim_{t\to\infty}\frac{\Big(E\big[\big(\eta(t)-\eta^{*}\big)^{2}\big]\Big)^{\frac{1}{2}}}{1-r_{1}}
\\
=&0.
\end{align*}

Therefore, by \eqref{eta-eta*} we have $$\lim_{t\to\infty}E\Big[\Big(\sum_{i=0}^{t-1}r_{1}^{i}\big(\eta(t-1-i)-\eta^{*}\big)\Big)^{2}\Big]=0.$$

Since $|r_{1}|<1$ and \eqref{xi1}, we have
\begin{align*}
&\lim_{t\to\infty}E\big[\big(\xi_{1}(t)-\frac{1}{1-r_{1}}\eta^{*}\big)^{2}\big]
\\
=&\lim_{t\to\infty}E\Big[\Big(r_{1}^{t}\xi_{1}(0)+\sum_{i=0}^{t-1}r_{1}^{i}\big(\eta(t-1-i)-\eta^{*}\big)
\\
&+\sum_{i=0}^{t-1}r_{1}^{i}\eta^{*}-\frac{1}{1-r_{1}}\eta^{*}\Big)^{2}\Big]
\\
=&\lim_{t\to\infty}E\Big[\Big(r_{1}^{t}\xi_{1}(0)+\sum_{i=0}^{t-1}r_{1}^{i}\big(\eta(t-1-i)-\eta^{*}\big)
-\frac{r_{1}^{t}\eta^{*}}{1-r_{1}}\Big)^{2}\Big]
\\
=&\lim_{t\to\infty}r_{1}^{2t}\xi_{1}^{2}(0)+2\lim_{t\to\infty}r_{1}^{t}\xi_{1}(0)E\big[\sum_{i=0}^{t-1}r_{1}^{i}\big(\eta(t-1-i)-\eta^{*}\big)\big]
\\
&+\lim_{t\to\infty}E\Big[\Big(\sum_{i=0}^{t-1}r_{1}^{i}\big(\eta(t-1-i)-\eta^{*}\big)\Big)^{2}\Big]
\\
&-2\lim_{t\to\infty}\frac{r_{1}^{t}\eta^{*}}{1-r_{1}}E\big[\sum_{i=0}^{t-1}r_{1}^{i}\big(\eta(t-1-i)-\eta^{*}\big)\big]
\\
&-2\lim_{t\to\infty}\frac{r_{1}^{2t}\eta^{*}\xi_{1}(0)}{1-r_{1}}+\lim_{t\to\infty}(\frac{r_{1}^{t}\eta^{*}}{1-r_{1}})^{2}
\\
=& 0.
\end{align*}

Thus, $\lim_{t\to\infty}E\big[\big(\xi_{1}(t)-\frac{1}{1-r_{1}}\eta^{*}\big)^{2}\big]=0$. Similarly, denoting $\xi_{i}(t)\triangleq\prod_{j=i+1}^{n-1}(\mathds{D}-r_{j})\xi(t)$, repeating the procedure, we have
$\lim_{t\to\infty}E\big[\big(\xi_{i}(t)-\frac{1}{\prod_{j=1}^{i}1-r_{j}}\eta^{*}\big)^{2}\big]=0$, for $i=1,\ldots,n-1$. By the definition of $\xi_{i}(t)$, we know that $\xi(t)=\xi_{n-1}(t)$. Thus,
$$
\lim_{t\to\infty}E\big[\big(\xi(t)-\xi^{*}\big)^{2}\big]=0,
$$
where $\xi^{*}=\frac{1}{\prod_{j=1}^{n-1}1-r_{j}}\eta^{*}$.

ii) By Part i), we know that $\xi_{i}(t)$ converges to $\frac{1}{\prod_{j=1}^{i}1-r_{j}}\eta^{*}$ in the mean square when $\eta(t)$ converges to $\eta^{*}$ in the mean square. Now we calculate the convergence rate of $\xi_{i}(t)$ under the condition that $\eta(t)$ converges to $\eta^{*}$ at the rate of $O(\frac{1}{t})$.

Firstly, we calculate $\big|\sum_{i=0}^{t-1}r_{1}^{i}\big(E\big[\big(\eta(t-1-i)-\eta^{*}\big)^{2}\big]\big)^{\frac{1}{2}}\big|$. Without loss of generality, assume that $r_{1}\ge 0$. Then, \eqref{t-i} is obtained.

Since $E[(\eta(t)-\eta^{*})^{2}]=O(\frac{1}{t})$, there exists $M_{\eta}>0$ such that $E[(\eta(t)-\eta^{*})^{2}]\le \frac{M_{\eta}}{t}$. Then, 
\begin{align*}\label{int-M}
&\frac{\sum_{i=0}^{t-1}r_{1}^{-i}\Big(E\big[\big(\eta(i)-\eta^{*}\big)^{2}\big]\Big)^{\frac{1}{2}}}{r_{1}^{1-t}}
\\
\le &\frac{\sum_{i=0}^{t-1}r_{1}^{-i}M_{\eta}\sqrt{i}}{r_{1}^{1-t}}.
\tag{B4}
\end{align*}

For sequence $r_{1}^{1-t}/\sqrt{t}$, it can be seen that $r_{1}^{1-t}/\sqrt{t}$ is strictly monotone and divergent when $t\ge -1/\ln r_{1}$.
Then, by Stolz-Ces\`aro theorem,
\begin{align*}
&\lim_{t\to\infty}\frac{\sum_{i=0}^{t-1}r_{1}^{-i}M_{\eta}\sqrt{i}}{r_{1}^{1-t}/\sqrt{t}}
\\
=&\lim_{t\to\infty}\frac{r_{1}^{-t}M_{\eta}/\sqrt{t}}{r_{1}^{-t}/\sqrt{t+1}-r_{1}^{1-t}/\sqrt{t}}
\\
=&\frac{\sqrt{M_{\eta}}}{1- r_{1}},
\end{align*}
which implies $\frac{\sum_{i=0}^{t-1}r_{1}^{-i}M_{\eta}\sqrt{i}}{r_{1}^{1-t}/\sqrt{t}}=O(1)$, or equivalently, $\frac{\sum_{i=0}^{t-1}r_{1}^{-i}M_{\eta}\sqrt{i}}{r_{1}^{1-t}}=O(\frac{1}{\sqrt{t}})$.  This together with \eqref{eta-eta*}-\eqref{int-M} gives
$$
E\Big[\Big(\sum_{i=0}^{t-1}r_{1}^{i}\big(\eta(t-1-i)-\eta^{*}\big)\Big)^{2}\Big]=O\Big(\frac{1}{t}\Big).
$$

Thus, $E\big[\big|\sum_{i=0}^{t-1}r_{1}^{i}\big(\eta(t-1-i)-\eta^{*}\big)\big|\big]=O(1/\sqrt{t})$ can also be obtained.
By \eqref{xi1}, we can get that
\begin{align*}
&E\big[\big(\xi_{1}(t)-\frac{1}{1-r_{1}}\eta^{*}\big)^{2}\big]
\\
=&r_{1}^{2t}\xi_{1}^{2}(0)+2r_{1}^{t}\xi_{1}(0)E\big[\sum_{i=0}^{t-1}r_{1}^{i}\big(\eta(t-1-i)-\eta^{*}\big)\big]
\\
&+E\Big[\Big(\sum_{i=0}^{t-1}r_{1}^{i}\big(\eta(t-1-i)-\eta^{*}\big)\Big)^{2}\Big]-2\frac{r_{1}^{2t}\eta^{*}\xi_{1}(0)}{1-r_{1}}
\\
&-2\frac{r_{1}^{t}\eta^{*}}{1-r_{1}}E\big[\sum_{i=0}^{t-1}r_{1}^{i}\big(\eta(t-1-i)-\eta^{*}\big)\big]+(\frac{r_{1}^{t}\eta^{*}}{1-r_{1}})^{2}
\\
=&O(r_{1}^{2t})+O\Big(\frac{r_{1}^{t}}{\sqrt{t}}\Big)+O\Big(\frac{1}{t}\Big)
\\
=&O\Big(\frac{1}{t}\Big).
\end{align*}

Similarly to the proof of Part i), repeating the procedure, we have
$$E\big[\big(\xi(t)-\xi^{*}\big)^{2}\big]=O\Big(\frac{1}{t}\Big).$$

\section{The proof in the switching topology case}

\subsection{The proof of Lemma \ref{Vi}}\label{app-Vi}
Repeating the analysis process in Appendix \ref{app-V}, we can conclude that 

\begin{align*}\label{delta(t)-t}
\delta(t)=&(I_{N}\otimes \tilde{A}-\frac{\gamma}{t}L_{m(t-1)}\otimes BK_{2})\delta(t-1)
\\
&+\frac{\gamma}{t}(J_{N}W_{m(t-1)}\otimes B)\hat{\varepsilon}(t-1),
\tag{C1}
\end{align*}
and
\begin{align*}\label{Vi-v1v2}
V(t)=&E[\hat{\delta}^{T}(t-1)(I_{N}-\frac{\gamma}{t}T_{G}^{-1}L_{m(t-1)}T_{G})^{2}\hat{\delta}(t-1)]
\\
&+\frac{2\gamma}{t}E[\hat{\delta}^{T}(t-1)(I_{N}-\frac{\gamma}{t}T_{G}^{-1}L_{m(t-1)}T_{G})
\\
&\cdot(T_{G}^{-1}J_{N}W_{m(t-1)})\hat{\varepsilon}(t-1)]+O\Big(\frac{1}{t^{2}}\Big).
\tag{C2}
\end{align*}
It can be seen that the only differences between \eqref{delta(t)} and \eqref{delta(t)-t}, \eqref{V-v1v2} and \eqref{Vi-v1v2}
are that the fixed matrices $L$ and $W$ have been modified into switching matrices $L_{m(t)}$ and $W_{m(t)}$.

Then, as Remark \ref{E[Lt]=L} says, by the property of conditional expectation and Lemma \ref{expectation}, we can get that $E[\hat{\delta}^{T}(t)L_{m(t)}\hat{\delta}(t)]=E[E[\hat{\delta}^{T}(t)L_{m(t)}\hat{\delta}(t)|\hat{\delta}(t)]]=E[\hat{\delta}^{T}(t)\check{L}\hat{\delta}(t)]+O(\lambda_{L}^{t})$, thus dealing with the switching matrix $L_{m(t)}$ in the following. To be specific, we have
\begin{align*}\label{v1-t}
V_{1}(t)=&E[\hat{\delta}^{T}(t-1)(I_{N}-\frac{\gamma}{t}T_{G}^{-1}L_{m(t-1)}T_{G})^{2}\hat{\delta}(t-1)]
\\
=&E[\hat{\delta}^{T}(t-1)(I_{N}-\frac{2\gamma}{t}T_{G}^{-1}L_{m(t-1)}T_{G})\hat{\delta}(t-1)]
\\
&+O\Big(\frac{1}{t^{2}}\Big)
\\
=&E\big[\hat{\delta}^{T}(t-1)\big(I_{N}-\frac{2\gamma}{t}T_{G}^{-1}\big(\check{L}+O(\lambda_{L}^{t})\big)T_{G}\big)\hat{\delta}(t-1)\big]
\\
&+O\Big(\frac{1}{t^{2}}\Big)
\\
=&E\big[\hat{\delta}^{T}(t-1)\big(I_{N}-\frac{2\gamma}{t}T_{G}^{-1}\check{L}T_{G}\big)\hat{\delta}(t-1)\big]+O\Big(\frac{1}{t^{2}}\Big),
\end{align*}
which is similar to the form of $V_{1}(t)$ in Appendix \ref{app-V}. 

At this point, the switching matrix $L_{m(t)}$ has been transformed into the fixed matrix $\check{L}$. 
By repeating the proof procedure from Lemma \ref{V} in Appendix \ref{app-V}, we can derive the following results.
\begin{align*}
V_{1}(t)
\le &\big(1-\frac{2\gamma\lambda_{2}}{t}\big)V(t-1)+O\Big(\frac{1}{t^{2}}\Big),
\tag{C3}
\end{align*}

\begin{align*}\label{v2-t}
V_{2}(t)
\le & \frac{\gamma}{t}\Big(\lambda_{2}\big(1-\frac{2\gamma\lambda_{2}}{t}\big)V(t-1)+\frac{\lambda_{G}}{\lambda_{2}}R(t-1)\Big)+O(\frac{1}{t^{2}}),
\tag{C4}
\end{align*}
where  the definition of $\lambda_{G}=\mathop{\max}\limits_{1\le i \le h}\{\|T_{G}^{-1}J_{N}W_{i}\|^{2}\}$ is different with the fixed topology case and $\lambda_{2}$ is the minimum non-negative eigenvalue of the union topology $G'$.

Combining \eqref{Vi-v1v2}-\eqref{v2-t}, we have
\begin{align*}
V(t)\le  \big(1-\frac{\gamma\lambda_{2}}{t}\big)V(t-1)+\frac{\gamma\lambda_{G}/\lambda_{2}}{t}R(t-1)+O\Big(\frac{1}{t^{2}}\Big),
\end{align*}
which is consistent with the fixed topology case.

\subsection{The proof of Lemma \ref{Ri}}\label{app-Ri}
Changing the fixed matrices $L, W$ into $L_{m(t)}, W_{m(t)}$ and $P_{m(t)}$ and repeating the analysis process in Appendix \ref{app-R}, we can conclude that
\begin{align*}\label{Ri-r1r2r3}
R(t)=& E[\|\hat{\varepsilon}(t)\|^{2}]
\\
=& E\Big[\Big\|\bm{\Pi_{M}}\Big\{\hat{z}(t-1)+\frac{\beta}{t}P_{m(t)}\big(\mathcal{F}(C-\hat{z}(t-1))
\\
&-s(t)\big)\Big\}-(Q\otimes K_{2})x(t)\Big\|^{2}\Big]
\\
=&E\Big[\hat{\varepsilon}^{T}(t-1)\big(I_{d}-\frac{\gamma}{t}QW_{m(t-1)}\big)^{T}\big(I_{d}-\frac{\gamma}{t}QW_{m(t-1)}\big)
\\
&\cdot\hat{\varepsilon}(t-1)\Big]+\frac{2\gamma}{t}E\Big[\hat{\varepsilon}^{T}(t-1)\big(I_{d}-\frac{\gamma}{t}QW_{m(t-1)}\big)^{T}Q
\\
&\cdot L_{m(t-1)} T_{G}\hat{\delta}(t-1)\Big]+\frac{2\beta}{t}E\Big[\hat{\varepsilon}^{T}(t-1)\big(I_{d}-\frac{\gamma}{t}Q
\\
&\cdot W_{m(t-1)}\big)^{T}P_{m(t)}\Big(\mathcal{F}\big(C-\hat{z}(t-1)\big)-s(t)\big)\Big)\Big]
\\
&+O\Big(\frac{1}{t^{2}}\Big).
\tag{C5}
\end{align*}
 Denote the first, second, and third items of \eqref{Ri-r1r2r3} as $R_{1}(t)$, $R_{2}(t)$, and $R_{3}(t)$, respectively, i.e., $R(t)\le R_{1}(t)+R_{2}(t)+R_{3}(t)+O\big(\frac{1}{t^{2}}\big)$. Then, we have
\begin{align*}\label{r1-t}
R_{1}(t)
&\le \big(1+\frac{\gamma\sqrt{\lambda_{QW}}}{t}\big)^{2}R(t-1),
\tag{C6}
\end{align*}
where $\lambda_{QW}=\mathop{\max}\limits_{1\le i \le h}\{\|QW_{i}\|^{2}\}$ is different with the fixed topology case.

Similarly to \eqref{r2}, we can get that
\begin{align*}\label{r2-t}
R_{2}(t)\le & \frac{\gamma}{t}\Big(\frac{\lambda_{QL}\lambda_{2}}{\lambda_{G}}\big(1+\frac{\gamma\sqrt{\lambda_{QW}}}{t}\big)^{2}R(t-1)+\frac{\lambda_{G}}{\lambda_{2}}V(t-1)\Big),
\tag{C7}
\end{align*}
where $\lambda_{QL}=\mathop{\max}\limits_{1\le i \le h}\{\|QL_{i}T_{G}\|^{2}\}$ and $\lambda_{2}$ are different with  \eqref{r2}.

Subsequently, using the conclusion in Appendix \ref{app-R}, since $\sum_{i=1}^{h}\pi_{i}P_{i}\ge \pi_{\min}I_{N}$, it can be seen that
\begin{align*}\label{r3-t}
R_{3}(t)=& -\frac{2\beta}{t}E\big[\hat{\varepsilon}^{T}(t-1)P_{m(t)}\text{diag}\big(\vec{f}(\zeta(t))\big)\hat{\varepsilon}(t-1)\big]
\\
&+O\Big(\frac{1}{t^{2}}\Big)
\\
=& -\frac{2\beta}{t}E\big[\hat{\varepsilon}^{T}(t-1)(\sum_{i=1}^{h}\pi_{i}P_{i})\text{diag}\big(\vec{f}(\zeta(t))\big)\hat{\varepsilon}(t-1)\big]
\\
&+O\Big(\frac{1}{t^{2}}\Big)
\\
\le &-\frac{2\beta f_{M}\pi_{\min}}{t}R(t-1)+O\Big(\frac{1}{t^{2}}\Big).
\tag{C8}
\end{align*}

Considering \eqref{Ri-r1r2r3} with \eqref{r1-t}-\eqref{r3-t}, we can obtain that
\begin{align*}
R(t)\le  &\Big(1-\frac{2\beta f_{M}\pi_{\min}-\gamma\alpha}{t}\Big)R(t-1)+\frac{\gamma\lambda_{G}/\lambda_{2}}{t}V(t-1)
\\
&+O\Big(\frac{1}{t^{2}}\Big),
\end{align*}
which has a new constant $\pi_{\min}$ that corresponding with the switching topologies.

\section*{References}
\bibliographystyle{ieeetr}
\bibliography{ref}

\begin{IEEEbiography}[{\includegraphics[width=1in,height=1.25in,clip,keepaspectratio]{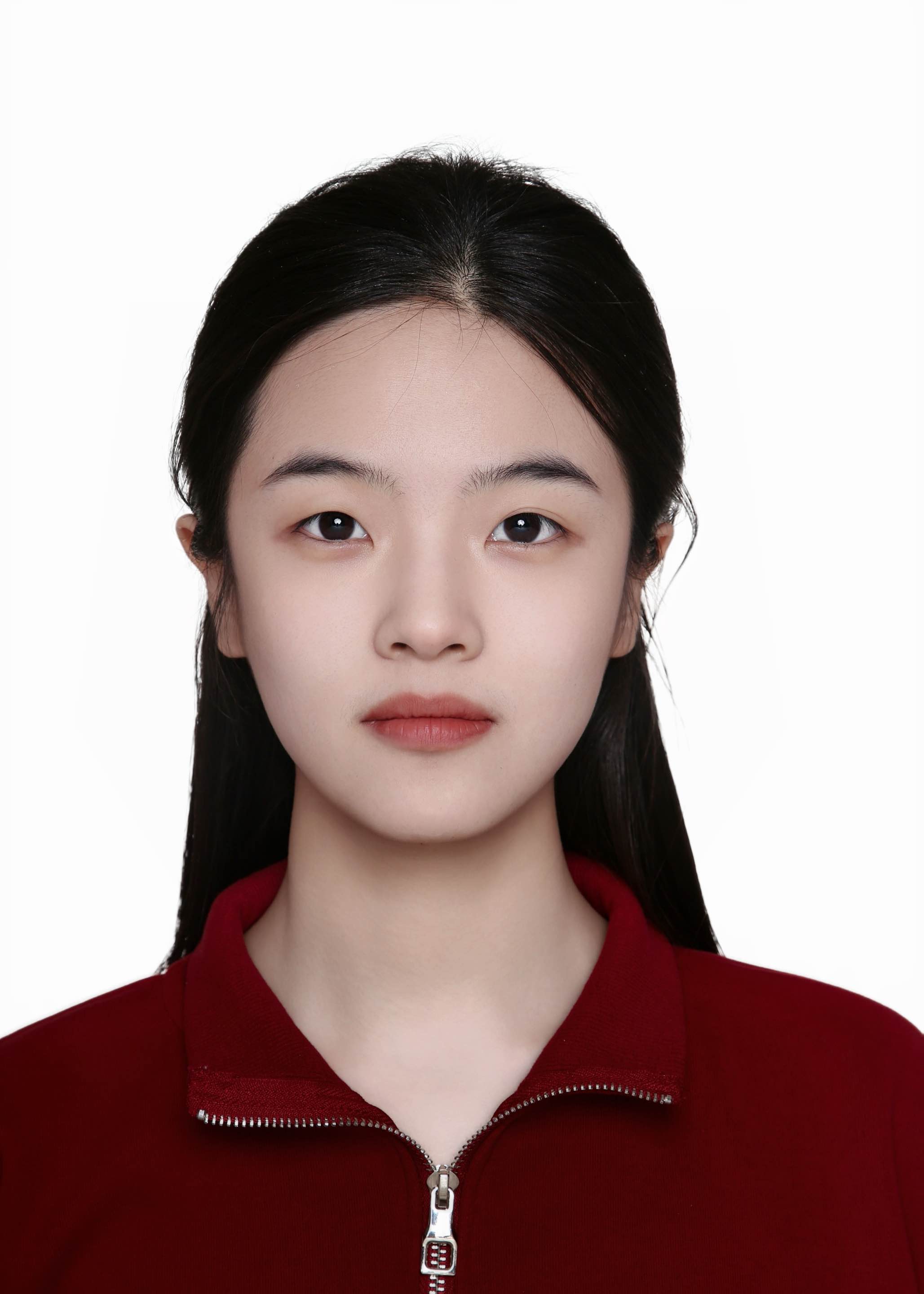}}]{Ru An}
received the B.S. degree in mathematics from Shandong University, Jinan, in 2022. She is currently working toward the Ph.D. degree majoring in system theory at the Academy of Mathematics and Systems Science (AMSS), Chinese Academy of Science (CAS), Beijing, China. 

Her research interests include quantized systems and multi-agent systems.
\end{IEEEbiography}

\begin{IEEEbiography}[{\includegraphics[width=1in,height=1.25in,clip,keepaspectratio]{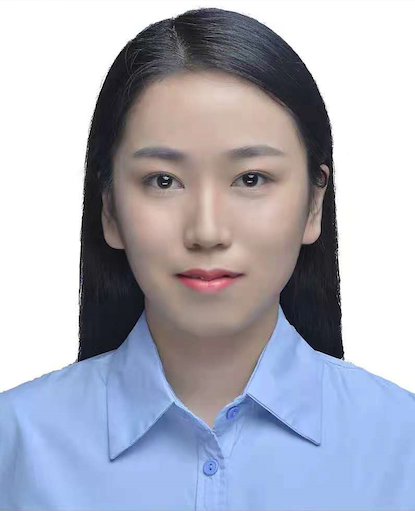}}]{Ying Wang}
received the B.S. degree in Mathematics from Wuhan University, Wuhan, China, in 2017, and the Ph.D. degree in system theory from the Academy of Mathematics and Systems Science (AMSS), Chinese Academy of Science (CAS), Beijing, China, in 2022. Now she is a postdoctoral fellow at the AMSS, CAS, and KTH Royal Institute of Technology, Stockholm, Sweden.

 Her research interests include identification and control of quantized systems, and multi-agent systems. \end{IEEEbiography}

\begin{IEEEbiography}[{\includegraphics[width=1in,height=1.25in,clip,keepaspectratio]{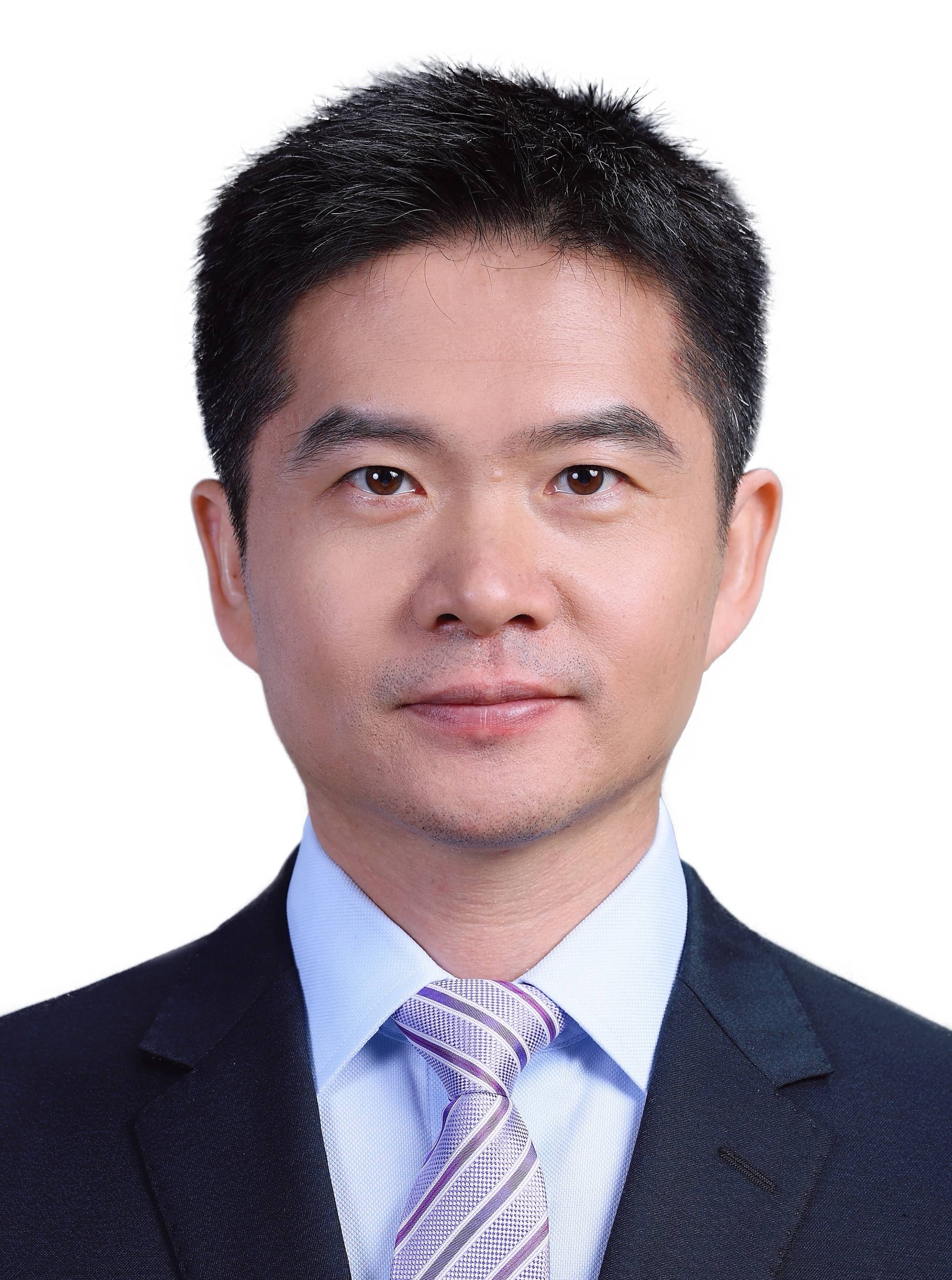}}]{Yanlong Zhao}
 received the B.S. degree in mathematics from Shandong University, Jinan, China, in 2002, and the Ph.D. degree in systems theory from the Academy of Mathematics and Systems Science (AMSS), Chinese Academy of Sciences (CAS), Beijing, China, in 2007. Since 2007, he has been with the AMSS, CAS, where he is currently a full Professor. 
 
His research interests include identification and control of quantized systems, information theory and modeling of financial systems.
 
He has been a Deputy Editor-in-Chief {\em Journal of Systems and Science and Complexity}, an Associate Editor of {\em Automatica}, {\em SIAM Journal on Control and Optimization}, and {\em IEEE Transactions on Systems, Man and Cybernetics: Systems}. He served as a Vice-President of Asian Control Association, and is now a Vice General Secretary of Chinese Association of Automation (CAA), a Vice-Chair of Technical Committee on Control Theory (TCCT), CAA, and a Vice-President of IEEE CSS Beijing Chapter.
\end{IEEEbiography}

\begin{IEEEbiography}[{\includegraphics[width=1in,height=1.25in,clip,keepaspectratio]{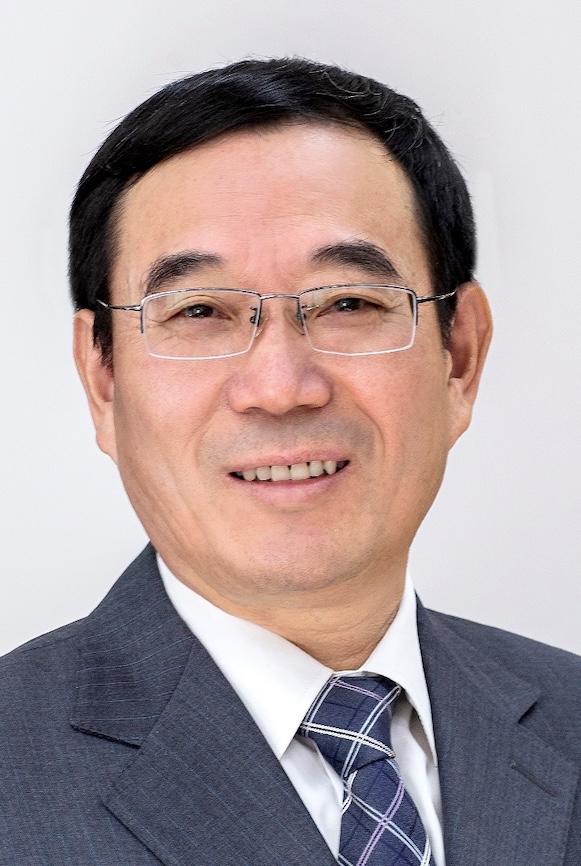}}]{Ji-Feng Zhang} 
received the B.S. degree in mathematics from Shandong University, China, in 1985, and the Ph.D. degree from the Institute of Systems Science, Chinese Academy of Sciences (CAS), China, in 1991. 
He is now with the Zhongyuan University of Technology and the Academy of Mathematics and Systems Science, CAS.

His current research interests include system modeling, adaptive control, stochastic systems, and multi-agent systems.

He is an IEEE Fellow, IFAC Fellow, CAA Fellow, SIAM Fellow, member of the European Academy of Sciences and Arts,  and Academician of the International Academy for Systems and Cybernetic Sciences. He received the Second Prize of the State Natural Science Award of China in 2010 and 2015, respectively. He was a Vice-Chair of the IFAC Technical Board, member of the Board of Governors, IEEE Control Systems Society; Convenor of Systems Science Discipline, Academic Degree Committee of the State Council of China; Vice-President of the Chinese Association of Automation, the Systems Engineering Society of China, and the Chinese Mathematical Society. He served as Editor-in-Chief, Deputy Editor-in-Chief or Associate Editor for more than 10 journals, including {\em Science China Information Sciences}, {\em IEEE Transactions on Automatic Control} and {\em SIAM Journal on Control and Optimization} etc.
\end{IEEEbiography}

\end{document}